\documentclass[11pt]{article}
\usepackage{amsmath}
\usepackage{amssymb}
\usepackage{amsthm}
\usepackage{amsfonts}
\usepackage{graphicx}
\usepackage{subfigure}

\usepackage{setspace}

\usepackage{bbm,epsfig,graphics,epic,color,rotating,color}
\numberwithin{equation}{section}

\newcommand{\R}{{\mathbb R}}
\newcommand{\Z}{{\mathbb Z}}
\newcommand{\TT}{{\mathbb T}}
\newcommand{\tb}{{\mathbf t}}
\newcommand{\rb}{{\mathbf r}}

\newcommand{\wt}[1]{\widetilde{#1}}
\newcommand{\wh}[1]{\widehat{#1}}

\newcommand{\reff}[1]{(\ref{#1})}
\newcommand{\E}{{\mathsf{E}}}

\newtheorem{theorem}{Theorem}[section]
\newtheorem{lemma}[theorem]{Lemma}
\newtheorem{proposition}[theorem]{Proposition}
\newtheorem{corollary}[theorem]{Corollary}

\newtheorem{definition}[theorem]{Definition}
\newtheorem{remark}[theorem]{Remark}

\theoremstyle{definition}

\title{Phase transition in ferromagnetic Ising model with a cell-board external field.}
\date{\today}
\author{Manuel Gonz\'alez-Navarrete\thanks{IME-USP, Rua do Mat\~{a}o, 1010, CEP 05508-090, S\~{a}o Paulo, Brazil. e-mail: manuelg@ime.usp.br}, \hspace{3pt} Eugene Pechersky\thanks{IITP, 19, Bolshoj Karetny per., Moscow, Russia. e-mail: pech@iitp.ru} \phantom{} and \phantom{a}Anatoly Yambartsev \thanks{IME-USP, Rua do Mat\~{a}o, 1010, CEP 05508-090, S\~{a}o Paulo, Brazil. e-mail: yambar@ime.usp.br}}

\begin{document}

\maketitle %

\vspace{-7pt}

\begin{abstract}
We show the presence of a first-order phase transition for a ferromagnetic Ising model on $\mathbb{Z}^2$ with a periodical external magnetic field. The external field takes two values $h$ and $-h$, where $h>0$. The sites associated with positive and negative values of external field form a cell-board configuration with rectangular cells of sides $L_1\times L_2$ sites, such that the total value of the external field is zero. The phase transition holds if $h<\frac{2J}{L_1}+ \frac{2J}{L_2}$, where $J$ is an interaction constant. We prove the first-order phase transition using the reflection positivity (RP) method. We apply a  key inequality which is usually referred to as the chessboard estimate. 

\end{abstract}

\section{Introduction}
\label{sec:zero}

In many models of statistical physics the phase transition is a result of spontaneous breaking of the symmetry of a system. The best known model with phase transition is the ferromagnetic Ising model (system) in the absence of a magnetic field. Essentially, this fact has been shown by Peierls \cite{Pei}. It has become a theorem by Griffiths \cite{Gr} and Dobrushin \cite{Dbr} (see also \cite{Si} and \cite{DiPP}). Peierls ideas are referred to as \textit{Peierls arguments} based on \textit{Peierls condition} and \textit{Peierls transformation}. Peierls condition means that  energy required  for a droplet formation of one of the phases surrounded by the sites  occupied by another phase is proportional to the size of the droplet boundary. For a two dimensional model (on $\Z^2$) the boundary size is the length of the droplet's boundary. The second component of Peierls arguments allows to perform Peierls transformation. It is based on a symmetry which  a studied model has. By Peierls transformation, it is possible to remove a  contour in the configuration such that only the energy of the contour is eliminated, and the energy of the rest part of the configuration is not changed. Peierls condition is unrelated to the model symmetry. Peierls condition is satisfied for the Ising model with a uniform external field, however there is no the symmetry in this case.  

Peierls arguments show a type of ``stability" of  ground states. It means that at a low temperature the state is ensemble of  small perturbations of the ground state which would result in a configuration ``close" to the starting ground state.

Unlike Peierls argument (specifically Peierls transformation), the \textit{Pirogov-Sinai theory} of phase transitions  allows one to find a low-temperature phase diagram of models with no symmetry requirement. When there is no a symmetry, the low-temperature  phase diagram is shifted with respect to the ground state diagram.

In addition, there exist a several more approaches. One such approach, \textit{Reflection Positivity} (RP), requires showing a type of reflection symmetry. Essentially, it is possible to prove a phase transition constructing a contour argument using the \textit{chessboard inequality} obtained from the RP property.

An external field added to the Hamiltonian can change the whole phase diagram. In the case of the ferromagnetic Ising model, any non-zero \underline{uniform} external field suppresses the phase transition. 
In some models where the magnetic field is not supposed to be uniform, it is possible to prove phase uniqueness, see for instance, \cite{Biss1}, \cite{Biss2}. A random external field can also suppress the phase transition in a planar Ising model (see \cite{Aiz1}, \cite{Aiz2}), even in the case when the total average of the external field is equal to 0.
 
In this paper we will address the problem of the existence of  phase transitions in a planar Ising model where the external field is periodic, forming a cell-board configuration such that total value of the magnetic field is zero. The initial motivation is coming from image processing  where Ising models with non-uniform external fields are used for analysing segmentation. The model in this study firstly were numerically studied by M.~Sigelle in \cite{Sig}. Reviews on the  applications of   Gibbs fields in  image processing can be found in Descombes and Zhizhina \cite{D&Z}   (see also \cite{MPS} and the book \cite{Win}). Posteriorly, Darbon and Sigelle \cite{D&S} proposed a grayscale fast and exact optimization method, it decomposing the target image in layers behaving as Ising models with cell-board external fields.

The models with staggered external fields can be useful in the theory of surfaces and domain theory of the solids.

In this work we consider the Ising model where the external field takes two values $h$ and $-h$, where $h>0$. The lattice $\mathbb{Z}^2$ is split into the union of disjointed cells of the same size, and the signs of the external field are alternated similar to a chessboard. Specifically, the cell with one sign of the external field is surrounded by four neighbor cells with the opposite value of the external field. We propose the reflection positivity method for the studies of this model. We will use a specific term for the alternated external field, {\it cell-board} partition, to avoid a confusion with the chessboard estimate, to be used later in the paper.

In \cite{NOZ} (by F.R. Nardi, E. Olivieri, and M. Zahradn\'{i}k) the authors study the case when the cells have infinite horizontal length, while their height equals one. Except the phase coexistence at low temperature, a lot of effort in \cite{NOZ} are focused on the proof of the uniqueness in a parameter region where the ground states coexist. 

Since our work is concentrated to the cases where Peierls condition is fulfilled then Pirogov-Sinai theory can be applied in this case. However, we used the reflection positivity method based on the periodicity of the cell-board external field. In this sense, as we will state in Corollary \ref{coro1}, a particular case of cell-board models  when the size of the cells is $1\times 1$, is trivially related with antiferromagnetic Ising model with uniform external field (see \cite{Dbr1}). For that model, in \cite{Fro2} the RP method has been used to prove the phase transition showed by Dobrushin \cite{DbrN}. Also RP property has been used to prove phase transition in planar rotor models with staggered external field (van Enter and Ruszel \cite{E&R}). In addition, Frohlich et al. \cite{Fro1} claimed that RP methods would produce better bounds of the critical temperature than Pirogov-Sinai approach can propose.

The paper is organized as follows: In sect. \ref{sec:general} we define the model we study, and present our main result (Theorem \ref{phase transition}). Sect. \ref{s3} contains a brief description of main ideas of the proof. The \textit{reflection positivity} technique, which is the main tool we use for the proofs, will be discussed in sect. \ref{sec:RP}, where we will also describe the \textit{chessboard estimates}. The proof of the main result using the RP technique follows the standard scheme, (see for example \cite{Bisk1}, Chapters 5, 6) and are given in sect. \ref{sec:proof}. The chessboard estimates as constructed in sect. \ref{sec:RP} does not encompass the external field in \cite{NOZ}, therefore in sect. \ref{sec:nardi} we study using RP a generalization of that model, again in the region with Peierls condition.

%%%%%%%%%%%%%%%%%%%%%%%%%%%%
%%%%%%%%%%%%%%%%%%%%%%%%%%%%
%%%%%%%%%%%%%%%%%%%%%%%%%%%%

\section{Definitions and results} 
\label{sec:general}

We {study} the ferromagnetic Ising model on $\mathbb{Z}^2$ with a periodic external field introduced in {\cite{Sig} (see also \cite{MPS})}.  Represent the lattice $\mathbb{Z}^2$ as the union of rectangular cells of the size $L_1\times L_2$, $L_i\in \mathbb N$: for each pair of integers $n, m$ we define
\begin{equation}
\begin{array}{ll}
\label{cellplus}
C(n,m) =  \bigl\{ (t_1,t_2) \in \Z^2 : & nL_1 \le t_1 < (n+1)L_1, \\[0.3cm]
& mL_2 \le t_2 < (m+1)L_2 \bigr\}.
\end{array}
\end{equation}
That is $\mathbb Z^ 2 = \cup_{n,m\in\mathbb Z} C(n,m)$. Let us define  subsets $\mathbf{Z}_+$ and $\mathbf{Z}_-$ of $\mathbb{Z}^2$:
\begin{equation}
\begin{array}{llr}
\label{Zplus}
\mathbf{Z}_+ = \displaystyle\bigcup_{\substack{n,m:\\ n+m \text{ is even}}} C(n,m), & & \hfill \mathbf{Z}_- = \Z^2 \setminus \mathbf{Z}_+.
\end{array}
\end{equation}
A site of $\Z^2$ is colored white if it is from $\mathbf{Z}_+$ and  black otherwise. Thus, the whole lattice is like a chessboard (see Figure~\ref{block}, where $L_1=3$ and $L_2=2$). 

Further we use a term {\it cell-board} since the term chessboard is used by reflection positivity technics which we apply.

Let $\Omega = \{-1,+1\}^{\Z^2}$ be the set of all configurations on $\Z^2$. The formal Hamiltonian is defined by 
\begin{equation}
\label{ham.chess}
H(\sigma)= -J \displaystyle\sum_{\langle t, s \rangle} \sigma(t) \sigma(s) - \displaystyle\sum_s h(s) \sigma(s),
\end{equation}
for any $\sigma \in \Omega$,
where $\sigma(t)\in\{-1,+1\}$ is a {\it spin value} of configuration $\sigma$ at the site $t\in \mathbb Z^2$, $J > 0$ is an interaction constant, a symbol $\langle t, s \rangle$ denotes unordered pairs of nearest neighbors $s, t \in \mathbb{Z}^2$, that is the Euclidean distance between the sites is one, $|t-s| = 1$, and an external field $h$ is given by 
\begin{equation}
\label{external}
h(s)=
\left\{
\begin{array}{rl}
h,&\text{ if }s\in \mathbf{Z}_+,\\
-h,&\text{ if }s\in \mathbf{Z}_-.
\end{array}
\right.
\end{equation}

Further, for any subset $\Lambda \subset \mathbb Z^2$ and any configuration $\sigma\in \Omega$, we will use the notation $\sigma(\Lambda)$ for the configuration of $\sigma$ restricted to the set of sites $\Lambda$.

We recall the standard definitions of a Gibbs field on the infinite lattice $\mathbb{Z}^2$ and related notations. Let $W$ be a finite subset from $\mathbb Z^2$, and let $\Omega_W$ be the set of all configurations on $W$: $\Omega_W=\{-1, 1\}^W$. The Gibbs probability of the configuration $\sigma \in \Omega_W$  with boundary conditions $\omega\in \Omega$, is given by

{\small \begin{equation}
\label{Gibbs1}
\mu_{\beta,W} (\sigma | \omega) = \frac{1}{Z_W(\beta)} \exp \Bigl( \beta J \displaystyle\sum_{\substack{\langle t, s \rangle:\\ t, s \in W}} \sigma(t) \sigma(s) +  \beta J \displaystyle\sum_{\substack{\langle t, s \rangle:\\ t \in W, s \notin W}} \sigma(t) \omega(s) + \beta \displaystyle\sum_{s \in W} h(s) \sigma(s) \Bigr),
\end{equation}}
where $\beta$ is a positive constant usually interpreted as the inverse temperature, and $Z_W(\beta)$ is a normalizing constant, called a \textit{partition function}. 

Let $\mathcal G_\beta$ be a set of all Gibbs states on $\Omega$ obtained by the thermodynamic limit.

A configuration $\tilde{\sigma}\in \Omega$ is a local perturbation of a configuration $\sigma\in \Omega$ if there exists a finite set $V\subset \mathbb{Z}^2$ such that

\begin{equation}
\label{perturbation}
\tilde{\sigma}(t)=
\left\{
\begin{array}{rl}
-\sigma(t),&\text{ if }t\in V,\\
\sigma(t),&\text{ if }t\notin V.
\end{array}
\right.
\end{equation}
A configuration $\sigma\in\Omega$ is called a \textit{ground state} for the Hamiltonian $H$, if for any local perturbation $\tilde{\sigma}$ of the configuration $\sigma$ the inequality
\begin{equation*}
H(\tilde{\sigma}) - H(\sigma) \ge 0,
\end{equation*}
is valid. Following \cite{Si} we say that the {\it Peierls condition} holds true, if there exists a positive constant $c_P > 0$ such that for any local perturbation $\tilde{\sigma}$ (as in \eqref{perturbation}) of a ground state $\sigma$ the inequality
\begin{equation}
\label{peierls1}
H(\tilde{\sigma}) - H(\sigma) \ge c_P |\partial V |,
\end{equation}
holds, where $\partial V=\{\langle t,s\rangle:\:t\in V,s\notin V\}$ is the boundary of the set $V$. The constant $c_P$ is called the \textit{Peierls constant}.

The following theorem provides the known results from  \cite{MPS} about the ground states and the Peierls condition for our model.
\medskip

\begin{theorem}
\label{pecherski}
If
\begin{equation}\label{Peierls2}
 h < \frac{2J}{L_1} + \frac{2J}{L_2},
 \end{equation}
 then there exist two periodical ground states, namely the constant configurations $\sigma^+ \equiv +1$ and $\sigma^- \equiv -1$. In addition, the Peierls condition holds, and the Peierls constant $c_P$ is equal to $2J - hL_1L_2/(L_1+L_2)$. If \reff{Peierls2} does not hold and
 \begin{equation}
 h > \frac{2J}{L_1} + \frac{2J}{L_2},
 \end{equation}
 then the configuration
 \begin{equation}
 \label{sigmacell}
\sigma_c(t)=
\left\{
\begin{array}{rl}
+1,&\text{ if }t\in \mathbf{Z}_+,\\
-1,&\text{ if }t\in \mathbf{Z}_-,
\end{array}
\right.
\end{equation}
is the unique periodic ground state.
\end{theorem}

\subsection{Main result. Phase transition for cell-board model}
 
The next theorem provides the presence of a first-order phase transition in the cell-board model.

\begin{theorem}
\label{phase transition}
Let the condition \reff{Peierls2} hold true, then there exists some $\beta_0 = \beta_0 (L_1,L_2)$, such that for any $\beta > \beta_0$, there exist two distinct measures $\mu^+_{\beta}$ and $\mu^-_{\beta} \in \mathcal{G}_{\beta}$, which satisfy

\begin{equation}
\label{theorem1}
\mu^{\pm}_{\beta} (\sigma(t) = \pm 1 ) > \frac{1}{2}.
\end{equation}
That means $|\mathcal{G}_{\beta}|>1$. Moreover
\begin{equation}\label{criticalB}
\beta_0 = \frac{8[(B_1B_2+4)\ln2+\ln(c(c+1))]}{2J-\frac{hL_1L_2}{L_1+L_2}},
\end{equation}
where $B_i,\ i=1,2$ are defined in \reff{blockside} and $c > 1$ is a combinatorial constant related to the number of contours of a given size.
\end{theorem}
\medskip

\begin{remark}:
\begin{itemize}
\item Estimates for $c$ can be found in \cite{L&M} and \cite{B&B}. In our case $c$ can be taken no greater than 9.
\item The denominator in \eqref{criticalB} is the Peierls constant defined in Theorem \ref{pecherski}.
\item Let $\beta_c$ the inverse critical  temperature, then $\beta_0 \geq \beta_c$.
\end{itemize}
\end{remark}

Theorem \ref{phase transition} is the main result in the paper. The proof is in sect.~\ref{sec:proof}. 
It is based on the reflection positivity machinery. We explain the reflection positivity (RP) technique in a way adapted to our model in the section \ref{sec:tool}.

We conclude this section with a well known fact about connection between a particular case of our model and the antiferromagnetic Ising model with a constant external field. The formal Hamiltonian for the antiferromagnetic model is
\begin{equation}\label{antif}
H_a(\sigma)= -J_a \displaystyle\sum_{\langle t, s \rangle} \sigma(t) \sigma(s) - \displaystyle h_a \sum_s \sigma(s),
\end{equation}
where the interaction constant $J_a$ is negative, $J_a<0$, that creates the antiferromagnetic interactions between the nearest spins, the external field $h_a$ is a real constant. The external field, $\pm h$, of the cell-board model, when $L_1=L_2=1$, should be equivalent to the antiferromagnetic model with the constant external field $h_a=h$. This fact has been discussed by Frohlich et al. \cite{Fro2} in the context of RP applications (see also \cite{Dbr1}). In our settings this result is the consequence of Theorem \ref{phase transition}.

\begin{corollary}\label{coro1} Let $J, h>0$.
If $h<4J$ and $\beta > \frac{2k}{4J-h}$ for some $k>0$, then the antiferromagnetic Ising model (\ref{antif}) with $J_a=-J$ and $h_a=h$ has two phases.
\end{corollary}

\begin{proof}
Consider the cell-board model with $L_1=L_2=1$, as defined in \eqref{ham.chess} and \eqref{external}. Now, define the transformation $\Psi$ of the configuration space $\Omega$, 
$\Psi:\:\Omega\to\Omega$,
\begin{equation*}
\Psi(\sigma)(t_1,t_2)=\begin{cases}\sigma(t_1,t_2),&\mbox{ if }t_1+t_2 \mbox{ even},\\
-\sigma(t_1,t_2), &\mbox{ otherwise,}
\end{cases}
\end{equation*} 
is an one-to-one transformation of  $\Omega$. Note that if in (\ref{antif}) we choose $J_a=-J$, where $J>0$, and $h_a = h$, then the transformation $\Psi$ does not change energy of the configurations and provides the direct equivalence of the models.
\end{proof}

%%%%%%%%%%%%%%%%%%%%%%%%%%%%
%%%%%%%%%%%%%%%%%%%%%%%%%%%%
%%%%%%%%%%%%%%%%%%%%%%%%%%%%

\section{Plan of the proof of Theorem \ref{phase transition}}\label{s3}  Our model has a set of reflection symmetries that allow us to apply the reflection positivity technique (see subsection \ref{sec:RP}). The proof of the RP property is Proposition \ref{ourisRP}. The reflections are with respect to lines parallel to the coordinate axes. Depending on the parity of sides $L_1$ and $L_2$, the reflecting lines can either go through the sites of $\Z^2$ or bisect edges of   $\Z^2$.  In our model, not all such lines are reflecting. As a result there are blocks of the sites in $\Z^2$ which do not have the reflection property, those blocks entirely reflected with respect to the reflecting lines. Definitions of the blocks see in \reff{lamb} and \reff{lambti}. 

We take a torus as a main scene of our considerations. The thermodynamical limit is corresponding to the growth of the torus size. 

Proving the main Theorem \ref{phase transition} we estimate the probability to have different spin values +1 and -1 at remote sites on the torus (see Proposition \ref{twopoints}). The goal is to show that this probability is small. It is clear that the event
$$
(\sigma(s)=+1,\sigma(t)=-1),
$$ 
when $s\ne t$, should generate Peierls contour which is a set of the edges having the different values on the edge ends. We use the contour arguments for the proof, however we have to use \textit{thick contours} (\textit{block contours}) consisted of the blocks.  The block contour is composed of the blocks in which  Peierls contour is passing. Any configuration on each such block takes the different spin values (a  \textit{bad block}). There is an exclusion which should be treated separately (see about \textit{double-blocks} in section \ref{sec:proof}). A small probability of the configurations on the bad block follows from the chessboard estimate (Theorem \ref{chess-board}) and from the Peierls condition (Lemma \ref{lem.hb}, see also Proposition \ref{prop1}). The chessboard estimate is applied to find an upper bound of the bad block probability, see \reff{4.39}. The Peierls condition is applied to make this upper bound small at small temperature.

%%%%%%%%%%%%%%%%%%%%%%%%%%%%
%%%%%%%%%%%%%%%%%%%%%%%%%%%%
%%%%%%%%%%%%%%%%%%%%%%%%%%%%

\section{A detailed plan of the proof. Constructions}
\label{sec:tool}
Together with the lattice $\Z^2$ we often  consider a graph
\begin{equation}\label{graphZ2}
\mathbb G=(\Z^2,\mathbb E),
\end{equation}
where $\mathbb E$ is a set of edges between the neighbouring sites. Along with the discrete spaces and sets we consider ``continuous" spaces (manifolds) as $\R^2$ and tori.

Now, we place the spin system on a two-dimensional torus. Let
\begin{equation}\label{torocont}
\wh{\TT}_N=\R^2/[(NL_1 \mathbb Z) \times(NL_2 \mathbb Z)],
\end{equation}
be a toric manifold. A map $\mathcal M_N:\R^2\to \wh{\TT}_N$, is such that for every rectangle
$$
\wh{\Delta}_{N,n}=\{\rb=(r_1,r_2):\:L_inN\leq r_i<L_i(n+1)N,\ i=1,2\},
$$
the restricted map $\mathcal M_N:\wh{\Delta}_{N,n}\to \wh{\TT}_N$, is a bijection. Let $\TT_N$ be an image of $\Z^2$ by $\mathcal M_N$
$$
\TT_N=\wh{\TT}_N\cap\mathcal M_N(\Z^2).
$$

The coordinate system of $\wh{\TT}_N$ is naturally induced by $\mathcal M_N$ from $\R^2$. An ambiguity because of multivaluedness of the map  $\mathcal M_N$ will not lead to confusions hereinafter.

We assume that $N$ is even. Then
$\Delta_N=\wh\Delta_{N,0}(0)\cap\Z^2$ is composed by $N^2/2$ cells of $\mathbf{Z}_+$ and the same amount of cells of type $\mathbf{Z}_-$. 

Let $\Omega_N = \{-1,+1\}^{\mathbb T_N}$ be the set of all configurations on the torus $\mathbb T_N$. We consider Hamiltonian $H_N$ with so called periodical boundary conditions: for any $\sigma \in \Omega_N$

\begin{equation}
\label{torusham}
H_N(\sigma)= -J \displaystyle\sum_{\langle t, s \rangle \in \mathbb{T}_N} \sigma(t) \sigma(s) - \displaystyle\sum_{s \in \mathbb{T}_N} h(s) \sigma(s).
\end{equation}
The Gibbs measure is
\begin{equation}
\label{Gibbs2}
\mu_{\beta,N} (\sigma) = \frac{1}{Z_N(\beta)} \exp \Bigl( \beta J \displaystyle\sum_{\substack{\langle t, s \rangle \in \mathbb{T}_N}} \sigma(t) \sigma(s)  + \beta \displaystyle\sum_{s \in \mathbb{T}_N} h(s) \sigma(s) \Bigr),
\end{equation}
where $Z_N(\beta)$ is the corresponding partition function:
\begin{equation}
\label{partfun}
Z_N(\beta) =  \displaystyle\sum_{\sigma\in\Omega_N} \exp \left( - \beta H_N(\sigma) \right).
\end{equation}

%%%%%%%%%%%%%%%%%%%%%%%%%%%%%%%%%%%%%%%%%%%%%%

\subsection{Reflection Positivity and chessboard estimate}
\label{sec:RP}

In this section we define the Reflection Positivity (RP) technique that we use. The main consequence of RP is the \textit{chessboard estimate}, which is used to prove phase coexistence in the models with RP property. This technique was developed in the works of Frohlich et al. \cite{Fro1,Fro2,Fro3,GJS}. Surveys about this method can be found in Georgii \cite{book2} and Shlosman \cite{Shl}.

We include some detailed explanations of the RP method, because in our case there exists the dependence of chessboard estimates on the size of the cells of the external field \eqref{external}. We will mainly use the notation and definitions of Biskup and Koteck\'y \cite{Bisk2} and Biskup \cite{Bisk1}.

\subsubsection{Reflection positivity.}

We define reflection symmetries with respect to lines orthogonal to one of the lattice directions. Assuming the lattice $\Z^2$ embedded in $\mathbb{R}^2$, we denote by $\Theta$ the group of all transformations of $\mathbb{R}^2$ generated by reflections of $\mathbb{R}^2$ with respect to lines orthogonal to one of the lattice directions  such that $\mathbb Z^ 2$ is invariant for any $\vartheta\in\Theta$: $\vartheta\Z^2 = \Z^2$. Let $\vartheta_P$ denote the reflection $\vartheta$ with respect to the line $P$. The group $\Theta$ is composed by  two distinct subgroups $\Theta^k$ $(k=0, 1/2)$, generated by reflections $\vartheta_{P_i^{(n,k)}}$ for which the corresponding lines are  
\begin{equation}\label{4.4}
P_i^{(n,k)}=\{ t=(t_1,t_2) \in \mathbb{R}^2: t_i = n + k \},
\end{equation}
for $i=1$ or $2$,  integer $n$ and $k=0$ or $1/2$. Reflections from $\Theta^ 0$ we will call the reflections \textit{through sites}: the corresponding reflection lines pass through the sites of $\Z^2$. Reflections from the set $\Theta^{1/2}$ we will call reflections \textit{through bonds}: the corresponding reflection lines bisect bonds of $\mathbb E$, \reff{graphZ2}.

The groups $\Theta^k, k=0,1/2,$ naturally generate  the reflections of the tori $\wh{\TT}_N$ and $\mathbb{T}_N$. Thus $\vartheta_P(\wh{\mathbb{T}}_N) = \wh{\mathbb{T}}_N$ and $\vartheta_P(\mathbb{T}_N) = \mathbb{T}_N$. 
The reflecting line $P$ in $\Z^2$ becomes two antipodal lines in the torus which splits the torus into two symmetric components, say $\TT_N^l$ and $\TT_N^r$, the \textit{left} and the \textit{right} halfs. We denote those lines with the same symbol $P$ as well as the reflection $\vartheta_P\in\Theta^k$ between the left and right halfs such that $\vartheta_P(\mathbb{T}^l_N)=\mathbb{T}^r_N$ and  vice versa (see Figure \ref{torus}).
Note that $\mathbb{T}^l_N \cap \mathbb{T}^r_N \in P$ for the reflections through the sites ($k=0$) and are disjoint for the reflections through the bonds ($k=1/2$).

\begin{figure}[ht]
\begin{center}
\epsfig{file=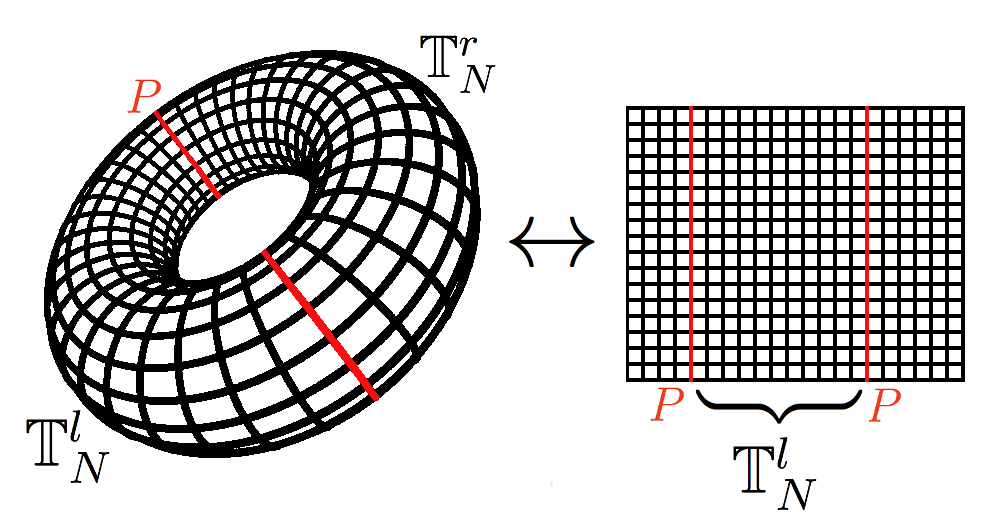, height=4cm, width=8cm}
\end{center}
\caption{A torus $\mathbb{T}_N$ is divided into the corresponding left and right half of the torus by reflecting line $P$ passing through sites, i.e. $\vartheta_P\in \Theta^0$.}\label{torus}
\end{figure}
\vspace{0.1cm}

Let $\mathcal{F}^l_P$ ($\mathcal{F}^r_P$) be a $\sigma$-algebra on $\Omega_N$ generated by all functions $\sigma(t), t\in \mathbb{T}^l_N\ (\mathbb{T}^r_N)$. As in \cite{Bisk2} we introduce a reflection operator $\theta_P :   \Omega_N \to  \Omega_N$: $\theta_P(\sigma(s)) : = \sigma(\vartheta_P(s))$ for any spatial reflection $\vartheta_P: \mathbb{T}^l_N \leftrightarrow \mathbb{T}^r_N$. The operator  $\theta_P$ obeys the following properties:

\begin{itemize}
\item [(1)]  $\theta_P$ is an involution, $\theta_P \circ \theta_P = id$;

\item [(2)] $\theta_P$ is a {\it reflection} in the sense that if $\mathcal{A} \in \mathcal{F}^l_P$ depends only on configurations on $\Lambda \subset \mathbb{T}^l_N$, then $\theta_P(\mathcal{A}) \in \mathcal{F}^r_P$ depends only on configurations on $\vartheta_P(\Lambda)$.
\end{itemize}

\begin{definition} 
\label{def.RP}
{\rm (Reflection Positivity \cite{Fro1,Fro2}, and see Definition 2.2 of \cite{Bisk2})}. Let $\mu$ be a probability measure on $ \Omega_N$, denote $\E_{\mu}$ the corresponding expectation, and let $P$ be a reflecting line. We say that $\mu$ is a \textrm{reflection positive} measure with respect to $\theta_P$ if for any two bounded $\mathcal{F}^l_P$-measurable functions $f$ and $g$  
\begin{equation}
\label{RP1}
\E_{\mu} (f \theta_P ( g )) = \E_{\mu} (g \theta_P (f)),
\end{equation}
and
\begin{equation}
\label{RP2}
\E_{\mu} (f \theta_P (f)) \ge 0,
\end{equation}
where $\theta_P ( f)$ is the $\mathcal{F}^r_P$-measurable function $f\circ \theta_P$.
\end{definition}

A consequence of RP is an inequality like the Cauchy-Schwarz inequality

\begin{equation}
\label{C-S}
[\E_{\mu} (f \theta_P( g)) ]^2 \le \E_{\mu} (f \theta_P( f)) \E_{\mu} (g \theta_P( g)).
\end{equation}

\subsubsection{Chessboard estimates.}

In this section, we recall the chessboard estimate in a form fitted to our case. The symmetries of $\TT_N$ which are used for the applications, are related to the symmetries of the external field. Since the external field is periodical any symmetry transformation should save block periods. The symmetry transformation of $\TT_N$ are reflections of ${\TT}_N$ with respect to lines in  $\wh\TT_N$. Let $\mathcal P= \mathcal P_1 \cup \mathcal P_2$ be a set of those lines being the union of the lines
\begin{equation}
\label{planes}
P_i^{(n)} = \{ t=(t_1,t_2)\in \mathbb R^2: \ t_i = nL_i + (L_i-1)/2\},\ i=1,2,
\end{equation}
where $\mathcal P_1=\{P_1^{(n)}\} $ and $\mathcal P_2=\{P_2^{(n)}\} $.

Note that if $L_i$ is odd then the corresponding reflection $\vartheta_{P_i^{(n)}}\in \Theta^0$, and if $L_i$ is even then the corresponding reflection $\vartheta_{P_i^{(n)}}\in \Theta^{1/2}$. Any such line cuts in half corresponding cells $C(n,m)$ (see \reff{cellplus}). The set of lines $\mathcal P$ provides the decomposition of $\mathbb{T}_N$ into rectangular blocks (see Figure~\ref{block}). In each block the total value of external field is equal to zero.

\begin{figure}[ht]
\begin{center}
\subfigure[]{
\epsfig{file=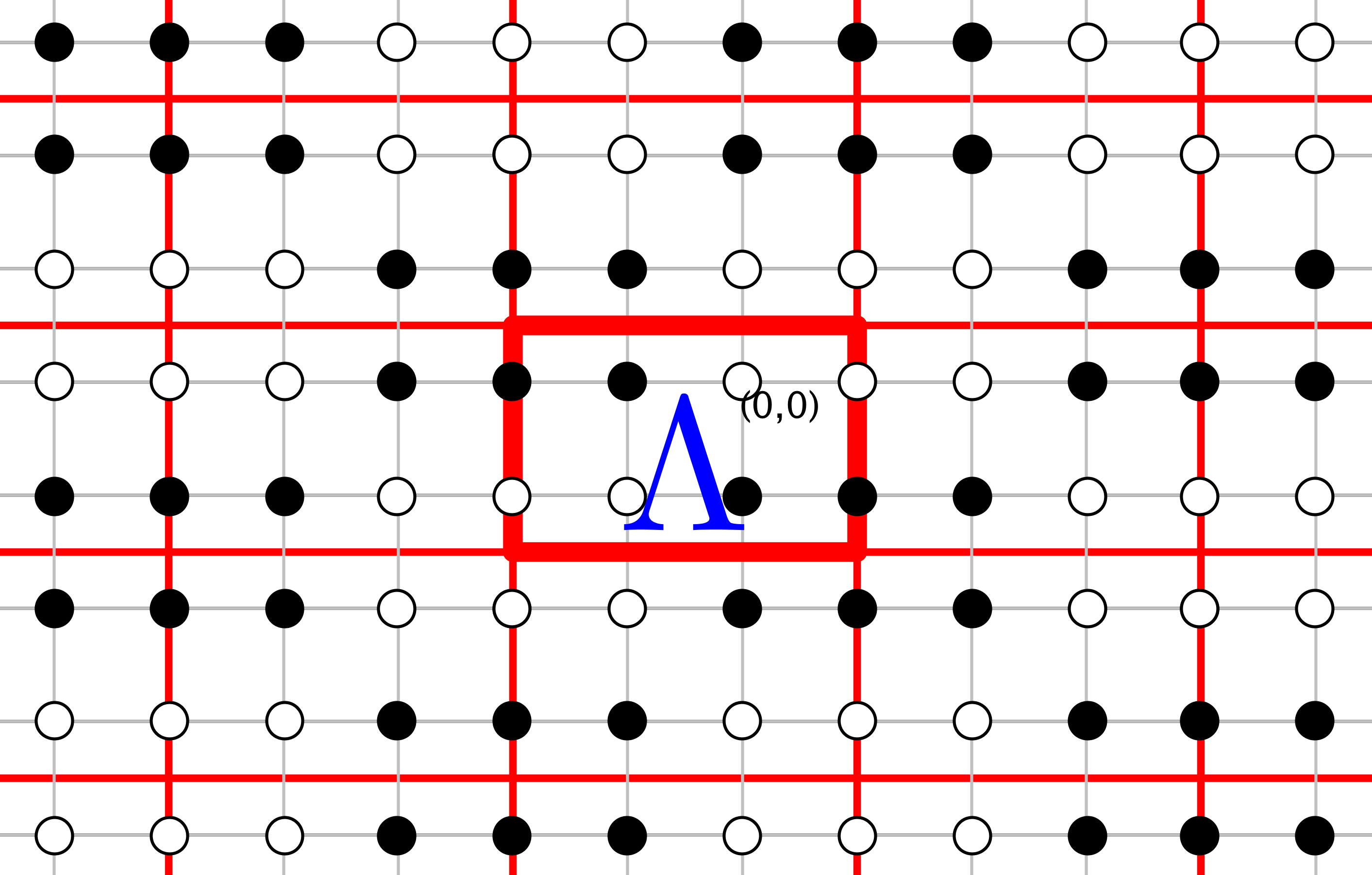, height=4cm}
\label{block}
}
\subfigure[]{
\epsfig{file=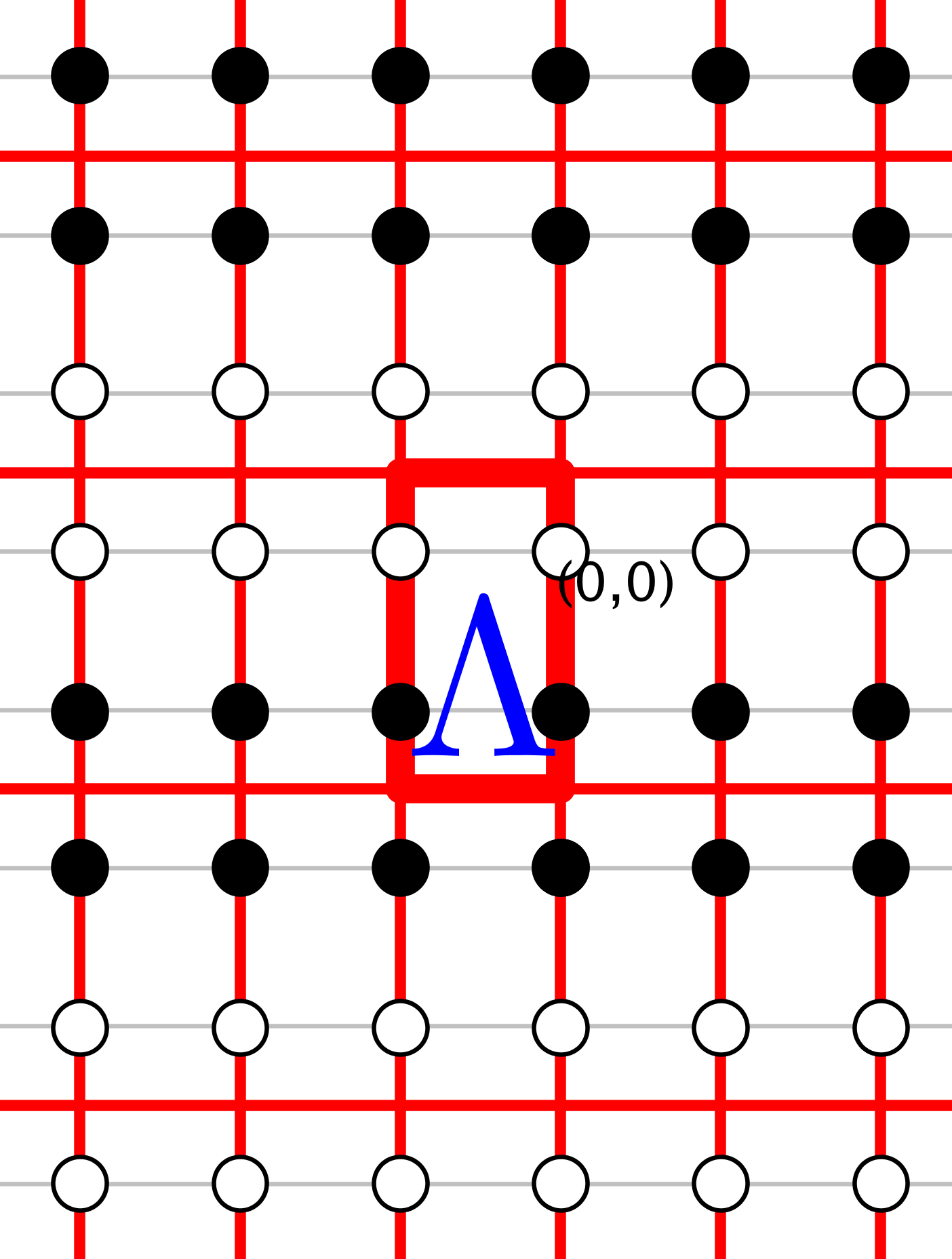, height=4cm}
\label{blocknardi}
}
\end{center}
\caption{Representation of the external field (black and white sites) and the ${\Lambda}$-blocks on the torus. Red lines indicate the set $\mathcal{P}$ of lines of reflection. (a) Cell-board model with $L_1=3$ and $L_2=2$. (b) The model studied in section \ref{sec:nardi} with $L=2$; it can be considered as the cell-board model with $L_1=\infty$ and $L_2=2$.}
\end{figure}
\vspace{0.1cm}

Let $\Lambda$ be the  minimal block, obtained with divisions by $\mathcal P$, which contains the origin, that is
\begin{equation}\label{lamb}
\Lambda=\Bigl\{ (t_1,t_2)\in\TT_N:\: \Bigl| t_1 + \frac{1}{2} \Bigr| \le \frac{L_1}{2},\ \Bigl| t_2 + \frac{1}{2} \Bigr| \le \frac{L_2}{2} \Bigr\}.
\end{equation}
A corresponding block $\wh{\Lambda}$ on $\wh{\TT}_N$ is 
\begin{equation}\label{lambti}
\wh{\Lambda}=\Bigl\{ (r_1,r_2)\in\wh{\TT}_N:\: \Bigl| r_1 + \frac{1}{2} \Bigr| \le \frac{L_1}{2},\ \Bigl| r_2 + \frac{1}{2} \Bigr| \le \frac{L_2}{2} \Bigr\}.
\end{equation} 

Note that the block $\Lambda$ contains $B_1 B_2$ sites, where
\begin{equation}
\label{blockside}
B_i=
\left\{
\begin{array}{rl}
L_i,&\text{ if }L_i \text{ is even},\\
L_i + 1,&\text{ if }L_i \text{ is odd}.
\end{array}
\right.
\end{equation}

The torus $\wh{\mathbb{T}}_N$ can be covered by translations of $\wh{\Lambda}$,
\begin{equation}
\label{mytil}
\wh{\mathbb{T}}_N = \bigcup_{\mathbf{r} \in \widetilde{\mathbb{T}}_N} (\wh{\Lambda} + \mathbf{r}),
\end{equation}
where $\widetilde{\mathbb{T}}_N = \{ \mathbf{r}=(r_1,r_2) \in \wh{\mathbb{T}}_N :r _1 = nL_1, r_2 = m L_2,\ n,m\in\Z \}$ is a qoutient subgroup of $\TT_N$. Correspondently the torus $\TT_N$ can be covered by sets $(\wh{\Lambda} + \mathbf{r})\cap\TT_N$:
$$
\mathbb{T}_N=\bigcup_{\mathbf{t} \in \widetilde{\mathbb{T}}_N} ({\Lambda} + \mathbf{t}).
$$

The neighboring translations of $\Lambda$ can have a side in common. Let $\Omega(\Lambda)=\{\sigma_\Lambda\}$ be the set of all configurations defined on $\Lambda$ and  let $\mathcal F_\Lambda$ be a $\sigma$-algebra of  events on $\Omega (\Lambda)$. We call $\Lambda$-\textit{event} the events $\mathcal A\in\mathcal F_\Lambda$.

Next we introduce some notions.  For each $s \in \mathbb{T}_N$, the map $\tau_s :  \Omega_N \to  \Omega_N$ is the {\it translation by $s$} defined {as} $(\tau_s \sigma)(t) = \sigma(t-s)$. We consider the lines

\begin{equation}\label{4.12}
Q_i=\left\{(t_1,t_2):\: t_i=-\frac12\right\},\ i=1,2,
\end{equation}
which bisect the block $\Lambda$. The reflections $\vartheta_{Q_1}$ and $\vartheta_{Q_2}$ are out of $\mathcal P$. In particular, these reflections do not shift $\Lambda$ and if $\sigma\in\Omega(\Lambda)$, the corresponding operators  do not preserve energy  
$$
H(\sigma)-H(\theta_{Q_i}(\sigma))=-2\sum_{t\in\Lambda}h(t)\sigma(t).
$$

A \textit{propagation} operator $\pi_\tb$ on $\Omega_N$ is defined with the help of two operators
\begin{equation}
j_\tb^{(i)}=\begin{cases}\label{4.14}
\theta_{Q_i},&\text{ if }t_i\text{ is odd},\\
\mathbf 1,&\text{ otherwise},
\end{cases}
\end{equation}
$i=1,2$, as following 
\begin{equation}\label{propag}
\pi_\tb(\sigma)=\tau_\tb\circ j_\tb^{(1)}\circ j_\tb^{(2)}(\sigma).
\end{equation}
Here $\tb=(t_1,t_2)\in\wt\TT_N$. 
The symbol $\mathbf 1$ above means the identical operator $\mathbf 1(\sigma)=\sigma$. The configuration $\sigma$ is shifted such that its values on $\Lambda$ is moved to $\Lambda+\tb$, with the possible reflections $\theta_{Q_1}$ and $\theta_{Q_2}$, depending on the parities of $t_1$ and $t_2$. The event $\pi_\tb(\mathcal A)$ is a cylindrical set of configurations from $\mathcal F_{\Lambda+\tb}$.

We remark that a propagator is based on reflection through the sides between two neighbors blocks. Let $\Pi_{\wt\TT_N}$ mean the set of all propagations corresponding the torus $\wt\TT_N$. Any propagation $\pi_\tb\in\Pi_{\wt\TT_N}$ is a bijection
$$
\sigma\in\Omega(\Lambda)\leftrightarrow\pi_\tb(\sigma)\in\Omega(\Lambda+\tb).
$$
Thus we can use the inverse map $\pi^{-1}_\tb$.

As we will see in the proof of Proposition \ref{twopoints}, depending on parity of $L_1$ and $L_2$, we work with four similar propagators $\pi^{*(k)}$, $k\in \{(h,1),(h,2),(v,1),(v,2)\}$, related to \textit{double-blocks} or $\Lambda^*$-\textit{blocks}. Particularly, when $L_1$ (resp. $L_2$) is even, we work with horizontal (resp. vertical) $\Lambda^*$-blocks, defined by

\begin{equation}
\label{alllambdas}
\begin{array}{lll}
\Lambda^*_{h,1}=\Lambda\cup\big(\Lambda+(L_1,0)\big), & & \Lambda^*_{h,2} =\Lambda\cup\big(\Lambda-(L_1,0)\big), \\ 
 \Lambda^*_{v,1}=\Lambda\cup\big(\Lambda+(0,L_2)\big),& &  \Lambda^*_{v,2}=\Lambda\cup\big(\Lambda-(0,L_2)\big).
\end{array}
\end{equation}
The associated quotient subgroups of $\TT_N$ are, respectively, given by

\begin{equation}
\begin{array}{lll}
\label{alltori}
{\wt\TT_N^{(h,1)}}=\{\tb=(t_1,t_2)\in{\wt\TT_N}:\:t_1=2nL_1,\ t_2=mL_2\},\\[0.2cm]
{\wt\TT_N^{(h,2)}}=\{\tb=(t_1,t_2)\in{\wt\TT_N}:\:t_1=(2n-1)L_1,\ t_2=mL_2\},\\[0.2cm]
{\wt\TT_N^{(v,1)}}=\{\tb=(t_1,t_2)\in{\wt\TT_N}:\:t_1=nL_1,\ t_2=2mL_2\},\\[0.2cm]
{\wt\TT_N^{(v,2)}}=\{\tb=(t_1,t_2)\in{\wt\TT_N}:\:t_1=nL_1,\ t_2=(2m-1)L_2\}.
\end{array}
 \end{equation}

Further we use a notion $\Lambda^*$-events for sets configurations defined on each of the $\Lambda_{k}^*$-blocks in \eqref{alllambdas}. Moreover, for brevity we denote 

\begin{equation}
\label{ddefinition}
\mathcal{D} :=\{(h,1),(h,2),(v,1),(v,2)\}.
\end{equation}

Now, we state the chessboard estimates.

\begin{theorem}{\rm (Chessboard estimate \cite{Fro1,Fro2,Bisk1,Shl})}
\label{chess-board}
Let $\mu_{\beta,N}$ a measure on $\Omega_N$ which is RP with respect to all reflections between the neighboring blocks $\Lambda + \mathbf{t}, \mathbf{t} \in \widetilde{\mathbb{T}}_N$. Then for any $\Lambda$-events $\mathcal{A}_1, \ldots, \mathcal{A}_m$ and any distinct sites $\mathbf{t}_1, \ldots, \mathbf{t}_m \in \widetilde{\mathbb{T}}_N$,
\begin{equation}
\label{chessboard1}
\displaystyle\mu_{\beta,N}\Bigl( \bigcap^m_{j=1} \pi_{\mathbf{t}_j} (\mathcal{A}_j) \Bigr) \le \displaystyle\prod_{j=1}^m \mu_{\beta,N}\Bigl(\bigcap_{\mathbf{t} \in \widetilde{\mathbb{T}}_N}  \pi_{\mathbf{t}} (\mathcal{A}_j)  \Bigr)^{1/ |\widetilde{\mathbb{T}}_N |}.
\end{equation}

Moreover, for any $\Lambda^*$-events $\mathcal{A}_1, \ldots, \mathcal{A}_m$ and any distinct sites $\mathbf{t}_1, \ldots, \mathbf{t}_m \in \widetilde{\mathbb{T}}_N^{(k)},\ k\in \mathcal{D}$,
\begin{equation}
\label{chessboard1_1}
\displaystyle\mu_{\beta,N}\Bigl( \bigcap^m_{j=1} \pi_{\mathbf{t}_j}^{*(k)} (\mathcal{A}_j) \Bigr) \le \displaystyle\prod_{j=1}^m \mu_{\beta,N}\Bigl(\bigcap_{\mathbf{t} \in \widetilde{\mathbb{T}}_N^{(k)}}  \pi_{\mathbf{t}}^{*(k)} (\mathcal{A}_j)  \Bigr)^{1/ |\widetilde{\mathbb{T}}_N^{(k)} |}.
\end{equation}
\end{theorem}

For the proof see, for example, \cite{Bisk1}, Theorem 5.8 or \cite{Bisk2}, Theorem 2.4.

The following quantities play the main role in the proof of the phase transition 
\begin{equation}
\label{zeta}
\mathfrak{z}_{\beta,N} (\mathcal{A} ) : = \mu_{\beta,N} \Bigl(\bigcap_{\mathbf{t} \in \widetilde{\mathbb{T}}_N}  \pi_{\mathbf{t}} (\mathcal{A})  \Bigr)^{1/ |\widetilde{\mathbb{T}}_N |},
\end{equation}
where $\mathcal{A}$ is $\Lambda$-event.
The function $\mathcal{A} \to \mathfrak{z}_{\beta,N}(\mathcal{A})$ is not  additive. However, given the $\sigma$-additivity of $\mu$ and using the chessboard estimate, it is easy to prove that it is sub-additive (see \cite{Bisk1}, Lemma 5.9). That is, for any collection of $\Lambda$-events $\mathcal{A}, \mathcal{A}_1, \mathcal{A}_2, \ldots$ such that
$\mathcal{A} \subset \cup_l \mathcal{A}_l$,
the inequality
\begin{equation}
\label{subz}
\mathfrak{z}_{\beta,N} (\mathcal{A} ) \le \sum_l \mathfrak{z}_{\beta,N} (\mathcal{A}_l)
\end{equation}
holds.
The limiting version of this quantity will be of particular interest for us. Thus, we define
\begin{equation}
\label{zetalim}
\mathfrak{z}_{\beta} (\mathcal{A} ) : = \lim_{N \to \infty} \mathfrak{z}_{\beta,N} (\mathcal{A} ).
\end{equation}
The existence of the limit follows from the sub-additivity.
Furthermore, we define for $\Lambda^*_{k}$-event $\mathcal{A}$, where $k\in \mathcal{D}$, the similar quantities 
\begin{equation}\label{zeta1_1}
\mathfrak{z}_{\beta,N}^{(k)}: = \mu_{\beta,N} \Bigl(\bigcap_{\mathbf{t} \in \widetilde{\mathbb{T}}_N^{(k)}}  \pi_{\mathbf{t}}^{*(k)} (\mathcal{A})  \Bigr)^{1/ |\widetilde{\mathbb{T}}_N^{(k)} |},
\end{equation}
and 
\begin{equation*}
\mathfrak{z}_{\beta}^{(k)} (\mathcal{A} ) : = \lim_{N \to \infty} \mathfrak{z}_{\beta,N}^{(k)}(\mathcal{A} ).
\end{equation*}
%%%%%%%%%%%%%%%%%%%%%%%%%%%%%%%%%%%%%%%%%%%%%%%%%%%%%%%%%%%%%%%%%%

\subsection{Phase coexistence}
\label{sec:proof}
The basis of the proof of Theorem~\ref{phase transition} are Propositions \ref{ourisRP}, \ref{prop1} %\ref{badblock} 
and \ref{twopoints} which we shall prove in the subsection \ref{remain}. 

The main applied technics  is  the reflection positivity technique. The proof
essentially consists on two steps. First, the easiest step, Proposition~\ref{ourisRP},  we apply a  known criterion for establish RP property for our model. 
Second,  we construct two measures $\mu_\beta^+$ and $\mu_\beta^-$ and prove that the probabilities $\mu_\beta^ + (\sigma(0) = -1)$ and $\mu_\beta^ - (\sigma(0) = +1)$ can be made less than $1/2$ for large $\beta$. This will prove the phase coexistence. In order to provide this we use the chessboard estimate \eqref{chessboard1} for a sort of Peierls arguments evaluating the contour probabilities. It is implemented in the proof of Proposition~\ref{twopoints}. 

The contour technique we work with are based on usage of thick contours partly assembled of a block set $\{\Lambda+\tb\}$, where $\tb\in\wt\TT_N$  and partly assembled of  the double-blocks $\{\Lambda^*\}$. We shall describe later all details.

We need some preliminary results about $\mathfrak{z}_{\beta, N}$ and $\mathfrak{z}_{\beta, N}^{(k)},\ k\in \mathcal{D}$ (see \reff{ddefinition}), defined in \reff{zeta} and \reff{zeta1_1}. 

\begin{proposition}
\label{ourisRP}
For any $P \in \mathcal{P}$ (see \reff{planes} and below), and all $\beta\ge 0$ the Gibbs measure $\mu_{\beta,N}$ \reff{Gibbs2} on the torus $\mathbb T_N$ is reflection positive (RP) with respect to $\theta_P$.
\end{proposition}
\noindent

In order to apply the chessboard inequality we introduce {\it bad block events} we deal with. Let $\sigma^+_{\Lambda}$ and $\sigma^-_{\Lambda}$ be the constant  configurations on $\Lambda$ with all spins plus and all spins minus, respectively.

For each configuration $\sigma_{\Lambda} \in  \{-1,+1\}^{\Lambda}$ on $\Lambda$ we define the event
\begin{equation}
\label{badconf}
\mathcal{B} ( \sigma_{\Lambda}) = \{ \sigma \in \Omega_N: \sigma(\Lambda) = \sigma_{\Lambda}\}.
\end{equation}

Let $R(\Lambda)$ be the set of all $\Lambda$-bad configurations, $R(\Lambda) = \{-1,+1\}^{\Lambda} \setminus \{ \sigma^+_{\Lambda}, \sigma^-_{\Lambda} \}$. Remember that the size of the block $\Lambda$ is equal to $B_1B_2$ sites, and $B_i \ge 2$ as defined in \eqref{blockside}. This implies that $|R(\Lambda)|=2^{B_1B_2} -2 \ge 14$.

Let $\mathcal{R}_\Lambda$ denote the event that the block $\Lambda$ is $\sigma$-bad for  $\sigma\in\Omega_N$, that is,

\begin{equation}
\label{setbadblock}
\mathcal{R}_\Lambda  = \left\{ \sigma \in \Omega_N: \sigma(\Lambda) \neq \sigma^{\pm}_{\Lambda}\right\} = \bigcup_{\sigma_{\Lambda} \in R(\Lambda)}\mathcal{B} ( \sigma_{\Lambda}) .
\end{equation}

The event $\mathcal{R}_\Lambda$ is called \textit{$\Lambda$-bad event}. Which represents all the torus configurations that are not constant on $\Lambda$-block.

\vspace{.3cm}

The proof of the main theorem about the phase coexistence is based on the contour technique. Both Peierls contours and thick contours are applied. The thick contours are consisted of the $\Lambda$-blocks and $\Lambda^*$-blocks. The $\Lambda$-block is included to the thick contour if Peierls contour touches this $\Lambda$-block, the $\Lambda^*$-block appears in the thick contour when Peierls contour at least partly passes between neighbouring $\Lambda$-blocks.  It happens  when  the size  $L_i$ in the direction of the block localisations is even.

If $\tb\in\{( L_1,0),(0, L_2)\}$ then the neighbouring $\Lambda$-blocks $\Lambda$ and $\Lambda+\tb$ are not intersected, that is  $\Lambda\cap(\Lambda+(L_1,0))=\emptyset$ $\bigl(\Lambda\cap(\Lambda+(0,L_2))=\emptyset \bigr)$ when  $L_1$ ($L_2$) is even. 

As in \reff{setbadblock} we define the $\Lambda^*$-bad-block events,
\begin{equation}
\mathcal{R}_{\Lambda^*_{k}} =  \left\{ \sigma \in \Omega_N: \sigma(\Lambda^*_{k}) \neq \sigma^{\pm}_{\Lambda^*}\right\},
\end{equation}
where $\sigma^+_{\Lambda^*}$ and $\sigma^-_{\Lambda^*}$ are the constant configurations on each of the $\Lambda^*_{k}$-blocks, $k \in \mathcal{D}$.

In the next proposition we show that the $\Lambda$-bad and $\Lambda^*$-bad events have small probability independently on $N$, when $\beta$ is large.
\begin{proposition}
\label{prop1}
If the condition \reff{Peierls2} holds true, then for any even $N$

\begin{equation}\label{4.28}
\mathfrak{z}_{\beta, N} (  \mathcal{R}_\Lambda) \le 2^{B_1B_2}\exp \left\{-\beta \Bigl(2J-\frac{hL_1L_2}{L_1+L_2} \Bigr) \right\}.
\end{equation}

Moreover, for any $k\in \mathcal{D}$, and any $N$ multiple of 4,
\begin{equation}\label{4.30}
\mathfrak{z}_{\beta, N}^{(k)} (  \mathcal{R}_{\Lambda^*_{k}}) \le 4^{B_1B_2}\exp \left\{-\beta \Bigl(2J-\frac{hL_1L_2}{L_1+L_2} \Bigr) \right\}.
\end{equation}
\end{proposition}

Now, we can state the main proposition.  
\begin{proposition}
\label{twopoints}
Let the condition \reff{Peierls2} holds true. There exists a constant $c >1$  such that for any $s, t \in \mathbb{T}_N$, the following inequality holds

{\small \begin{equation}
 \mu_{\beta,N} \left( \sigma(s) = +1, \sigma(t) = -1 \right) \le 2 c(c+1) 2^{B_1B_2 / 2}\exp \left\{-\frac{\beta}{8} \Bigl(2J-\frac{hL_1L_2}{L_1+L_2} \Bigr) \right\},
\end{equation}
}
for any
{\small \begin{equation}\label{betaP3}
\beta>\beta'=\frac{4[(B_1B_2+2)\ln 2+2\ln(c(c+1))]}{2J-\frac{hL_1L_2}{L_1+L_2}}.
\end{equation}}
\end{proposition}

The constant $c$ appears from the combinatorial argument related to the number of contours builded from $n$ $\Lambda$- and $\Lambda^*$-blocks.

\subsubsection{Proof of Theorem \ref{phase transition}.}

First of all, we use the following symmetry of the torus measure. Let ${\Lambda}_s$ be the block containing the site $s \in \mathbb{T}_N$, then
\begin{equation}
\label{prova1a}
\mu_{\beta,N} \left(\sigma: \sigma(\Lambda_s) \equiv +1 \right)   = \mu_{\beta,N} \left(\sigma: \sigma(\Lambda_s) \equiv -1 \right) = \frac{1 - \mu_{\beta,N} ( \mathcal{R}(\Lambda_s))}{2},
\end{equation}
for any $s \in \mathbb{T}_N$. In order to check this equality we apply the following two transformations for any configuration $\sigma \in \Omega_N$, such that $\sigma(\Lambda) \equiv +1$. First, we apply on $\sigma$ the reflection operator $\theta_{Q_1}$, defined as in section \ref{sec:RP}, where $Q_1$ is given by \reff{4.12}. That is, $\omega := \theta_{Q_1}(\sigma) \in \Omega_N$ takes the value $\omega(t_1,t_2)= \sigma(-t_1-1,t_2)$, for all $t=(t_1,t_2) \in \mathbb{T}_N$. Second, we obtain $\sigma^{\prime}= - \omega$, flipping all the spin values. In other words, $\sigma^{\prime}(t)= - \sigma(-t_1-1,t_2)$, for all $t \in \mathbb{T}_N$. Clearly, $\sigma^{\prime}(\Lambda) \equiv -1$ and given $h_{(t_1,t_2)} = - h_{(-t_1- 1,t_2)}$, the Hamiltonians are equal $H_N(\sigma) = H_N (\sigma^{\prime})$. It proves \reff{prova1a}.

Since the symmetry property of the model the following equalities hold
$$
\mu_{\beta,N}(\mathcal R_\Lambda)=\mu_{\beta,N}(\mathcal R_{\Lambda+\tb}),
$$
for any $\tb\in\wt{\TT}_N$. Therefore we omit sometimes the index $\Lambda$ at $\mathcal R$.

Using the chessboard estimate \eqref{chessboard1}, we obtain the inequality

\begin{equation}
\label{4.39}
 \mu_{\beta,N} ( \mathcal{R}) \le \mu_{\beta,N} \Bigl(\bigcap_{\mathbf{t} \in \widetilde{\mathbb{T}}_N}  \pi_{\mathbf{t}} (\mathcal{R})  \Bigr)^{1/ |\widetilde{\mathbb{T}}_N |} = \mathfrak{z}_{\beta,N} (\mathcal{R} ).
\end{equation}
Then, from \eqref{prova1a}
\begin{equation}
\label{prova1b}
\mu_{\beta,N} \left(\sigma: \sigma(\Lambda_s) \equiv +1 \right)  \ge \frac{1 - \mathfrak{z}_{\beta,N}(\mathcal{R})}{2}.
\end{equation}
Let $t \in \mathbb{T}_N$ such that $t = s + (NL_1/2,0)$, and define
\begin{equation}
\label{measures}
\mu_{\beta,N}^{\pm} ( \cdot )  : = \mu_{\beta,N} ( \cdot \mid \sigma(t) = \pm 1) .
\end{equation}
By \eqref{prova1b}, \reff{4.28}, and Proposition \ref{twopoints} we have
\begin{equation}\label{muplus}
\begin{array}{ll}
\mu_{\beta,N}^+ \left( \sigma(s) = -1 \right)   & \le \displaystyle\frac{\mu_{\beta,N} ( \sigma(s) = -1 , \sigma(t) = + 1)}{\mu_{\beta,N} \left(\sigma: \sigma(\Lambda_t) \equiv +1 \right)} \\[0.4cm]
& \le\displaystyle\frac{4c(c+1) 2^{B_1B_2/2}\exp\left\{-\frac{\beta}{8}\left(2J-h\frac{L_1L_2}{L_1+L_2}\right)\right\} }{1 - 2^{B_1B_2}\exp\left\{-\beta\left(2J-h\frac{L_1L_2}{L_1+L_2}\right)\right\}},
\end{array}
\end{equation}
and
\begin{equation}
\label{muminus}
\mu_{\beta,N}^- \left( \sigma(s) = +1 \right)  \le\displaystyle\frac{4c(c+1) 2^{B_1B_2/2}\exp\left\{-\frac{\beta}{8}\left(2J-h\frac{L_1L_2}{L_1+L_2}\right)\right\} }{1 - 2^{B_1B_2}\exp\left\{-\beta\left(2J-h\frac{L_1L_2}{L_1+L_2}\right)\right\}}.
\end{equation}

When  $N\nearrow\infty$ we extract from the sequences of the measures $(\mu_{\beta,N}^+)$ and $(\mu_{\beta,N}^-)$, two converging subsequences. Let  $\mu^+_{\beta}$ and $\mu^-_{\beta}$ be corresponding limits. Those measures are infinite-volume Gibbs measures corresponding to the Hamiltonian $H$ (\reff{ham.chess} and \reff{external}). It follows from DLR-equation that those measures are Gibbsian (see \cite{Bisk1}).

By \eqref{muplus} and \eqref{muminus} the inequalities \eqref{theorem1} are satisfied if
\begin{equation}\label{meq}
16c(c+1) 2^{B_1B_2}\exp\left\{-\frac{\beta}{8}\left(2J-h\frac{L_1L_2}{L_1+L_2}\right)\right\} < 1.
\end{equation}

This inequality means that the phase transition holds for all
\begin{equation}\label{4.50}
\beta > \frac{8[(B_1B_2+4)\ln2+\ln(c(c+1))]}{2J-\frac{hL_1L_2}{L_1+L_2}}.
\end{equation}
\qed

\subsection{Remaining proofs}\label{remain}

\subsubsection{Proof of Proposition~\ref{ourisRP}.}

The proof is the application of the known criteria for a measure to be reflection positive.  Fix a line $P \in \mathcal{P}$ of the reflections and let $\theta_P$ be the corresponding reflection operator. The criteria applied to our case claims that the measure $\mu_{\beta,N}$ is reflection positive, if its Hamiltonian can be represented in the form
\begin{equation}
\label{criterion}
-H_N = A + \theta_P( A) + \sum_{\alpha} C_{\alpha} \theta_P (C_{\alpha}),
\end{equation}
where $A, C_{\alpha}$ are $\mathcal{F}^l_P$-measurable functions. Then for all $\beta \ge 0$ the torus Gibbs measure $\mu_{\beta,N}$, is RP with respect to $\theta_P$ (see Definition \ref{def.RP}). The criteria can be found in Theorem 2.1 of Shlosman \cite{Shl} or Corollary 5.4 of Biskup \cite{Bisk1}.

In our case there are two possibilities for $P \in \mathcal{P}$: $P$ passes trough sites of $\mathbb{T}_N$ or not. In the case of $P$ passing through the sites of $\mathbb{T}_N$ choose 

{\small \begin{eqnarray*}
A &=& J \sum_{ \substack{\langle t, s \rangle: \\  t \in \mathbb{T}^l_N, s \in \mathbb{T}^l_N\setminus P}} \sigma(t) \sigma(s) + \frac J2 \displaystyle\sum_{\langle t, s \rangle \in P} \sigma(t) \sigma(s)+ \sum_{s \in \mathbb{T}^l_N \setminus P} h(s) \sigma(s)+ \frac12\sum_{s \in P} h(s) \sigma(s),
\end{eqnarray*}}
then, since $h(s)=h(\vartheta_P(s))$, 
\begin{equation*}
-H_N(\sigma) = A + \theta_P(A),
\end{equation*}
here functions $C_\alpha\equiv 0$. In the case of reflections through the bonds choose 
\begin{eqnarray*}
A &=& J \displaystyle\sum_{ \langle t, s \rangle \in \mathbb{T}^l_N} \sigma(t) \sigma(s) + \displaystyle\sum_{s \in \mathbb{T}^l_N} h(s) \sigma(s),
\end{eqnarray*}
then
\begin{equation*}
-H_N (\sigma) =   A + \theta_P (A)
+ J\displaystyle\sum_{\substack{t\in \mathbb{T}^l_N:\\ |t-P|=1/2}}\sigma(t)\theta_P(\sigma(t)).
\end{equation*}
The equality $h(s)=h(\vartheta_P(s))$ is used again.
That proves the proposition. \qed

\subsubsection{Proof of Proposition~\ref{prop1}.}

For simplicity, we only prove \eqref{4.28}, but the proof for each $\Lambda^*_{k}$-bad-event, $k \in \mathcal{D}$ is the same. We justify the condition $N$ multiple of $4$ in \eqref{4.30}, since the number of $\Lambda$-blocks fulfilling the whole torus $\TT_N$, is twice the number of $\Lambda^*$-blocks needed.

Let $\sigma_{{\Lambda},N} : = \cap_{\mathbf{t} \in \widetilde{\mathbb{T}}_N}  \pi_{\mathbf{t}} (  \mathcal{B} ( \sigma_{\Lambda})  )$ be the configuration on $\mathbb{T}_N$, obtained by the propagations of a fixed block configuration $\sigma_{\Lambda}$. 
The proof of the proposition~\ref{prop1} will be based on  the following inequality

{\small \begin{equation}
\label{eqprop2}
\mathfrak{z}_{\beta,N} ( \mathcal{B} ( \sigma_{\Lambda}) )^{| \widetilde{\mathbb{T}}_N |} = \frac{\exp(-\beta H_N(\sigma_{\Lambda,N}))}{Z_N(\beta)} \le \exp\bigl( -\beta \left[H_N(\sigma_{\Lambda,N}) - H_N(\sigma^+)\right] \bigr).
\end{equation}}
A bound of the right hand side  of \reff{eqprop2} can be found from the next two lemmas. 

In order to formulate the first lemma we introduce some notions. Consider the configuration $\sigma_{\Lambda,N}$ defined above. For any $\sigma_\Lambda$ the configuration $\sigma_{\Lambda,N}$ has the following periodicity property: for any $t \in \mathbb{T}_N$ and $\mathbf{t} \in \widetilde{\mathbb{T}}_N$ we have
\begin{equation}
\label{periodicity1}
\sigma_{{\Lambda},N} ( t ) = \sigma_{{\Lambda},N}(t + 2 \mathbf{t}).
\end{equation}
It means that there exists some ``minimal" sublattice $\Lambda^{[2\times 2]}$ of $\TT_N$, such that the configuration $\sigma_{\Lambda,N}$ can be obtained by  translations of $\sigma_{\Lambda,N}(\Lambda^{[2\times 2]})$. Indeed, using the rectangle $\wh\Lambda$ defined by \reff{lambti} let us define
{\small \begin{eqnarray}
\wh\Lambda^{[2\times 2]} &:=& 
 \wh\Lambda \cup \left(\wh\Lambda + (0, L_2) \right) \cup 
\left( \wh\Lambda + (L_1, 0) \right) \cup \left(\wh\Lambda + (L_1, L_2) \right)
+ \left( \frac{1}{4}, \frac{1}{4} \right),  \nonumber \\
\Lambda^{[2\times 2]} &:=& \mathbb{T}_N \cap \wh\Lambda^{[2\times 2]}. \label{torus_2}
\end{eqnarray}}
See Figure~\ref{blocktorus} for illustration.
\begin{figure}[ht]
\begin{center}
\epsfig{file=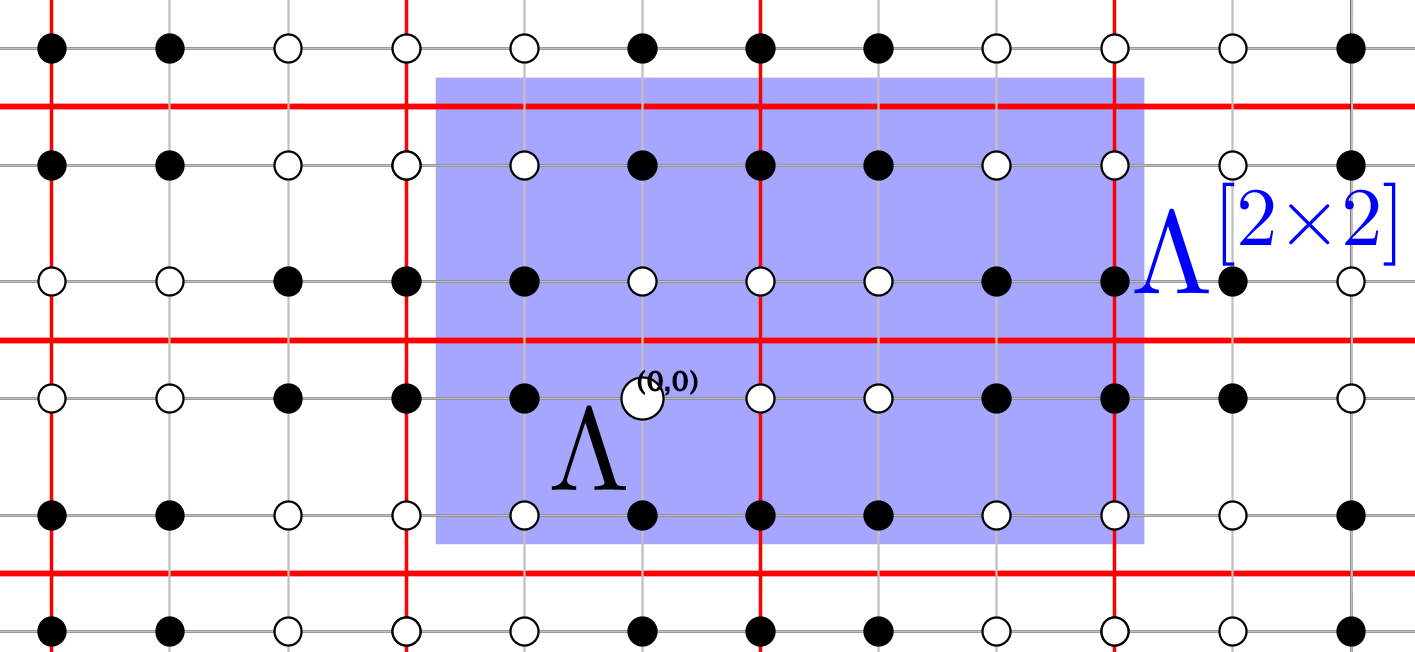, height=3.5cm, width=6.5cm}
\end{center}
\caption{The sublattice $\Lambda^{[2\times2]}$ is composed by the sites inside the blue rectangle $\wh\Lambda^{[2\times2]}$. Here $L_1=3$ and $L_2=2$.}\label{blocktorus}
\end{figure}

Remark that $(\Lambda^{[2\times 2]} +  2\mathbf{t}_1) \cap (\Lambda^{[2\times 2]} +  2\mathbf{t}_2) = \o$, if $\mathbf{t}_1 \neq \mathbf{t}_2$, with $\mathbf{t}_1, \mathbf{t}_2 \in \widetilde{\mathbb{T}}_N$. For any $\mathbf{t}\in \widetilde{\mathbb{T}}_N$, we have $h(s) = h(s+2\mathbf{t})$, as well as,

{\small \begin{equation}
\label{periodicity2}
\sigma_{{\Lambda},[2\times 2]} : =\sigma_{{\Lambda},N} (\Lambda^{[2\times 2]}) =  \sigma_{{\Lambda},N} (\Lambda^{[2\times 2]} + 2\mathbf{t})\mbox{  and  }
\mathbb{T}_N = \bigcup_{\mathbf{t} \in \widetilde{\mathbb{T}}_N} (\Lambda^{[2\times 2]} + 2\mathbf{t}).
\end{equation}}

Let $\widetilde{\TT}^{[2\times 2]}_N = \{ \tb = ( t_1, t_2) \in \mathbb{T}_N : t_1 = 2nL_1, t_2= 2mL_2,\ n,m\in\Z \}$ and $\TT^{[2\times 2]}=\TT_N/\wt\TT^{[2\times 2]}_N$.
Let us define the torus Hamiltonian: for any $\sigma \in \{-1,+1\}^{\mathbb{T}^{[2\times 2]}}$
\begin{equation}
\label{torus_2ham}
H_{[2\times 2]}(\sigma)= -J \displaystyle\sum_{\langle t, s \rangle \in \TT^{[2\times 2]} } \sigma(t) \sigma(s) - \displaystyle\sum_{s \in \TT^{[2\times 2]} }h(s) \sigma(s).
\end{equation}
The sites $\left(-\lfloor \frac{2L_1+1}4\rfloor,t_2\right)$\footnote{$\lfloor x \rfloor$, denotes the floor function, that is: $\lfloor x \rfloor=\max\, \{m\in\mathbb{Z}\mid m\le x\}$.} and $\left(\lfloor \frac{3L_1-1}2\rfloor,t_2\right)\in\Lambda^{[2\times 2]}$, are neighbours in $\TT^{[2\times 2]}$ for any $t_2\in\left[-\lfloor \frac{2L_2+1}{4}\rfloor,\lfloor \frac{3L_2-1}2\rfloor\right]$. As well as, $\left(t_1,-\lfloor\frac{2L_2+1}4\rfloor\right)$ and $\left(t_1,\lfloor\frac{3L_2-1}2\rfloor\right)\in\Lambda^{[2\times 2]}$, are neighbours when $t_1\in\left[-\lfloor\frac{2L_1+1}4\rfloor,\lfloor\frac{3L_1-1}2\rfloor \right]$. 
\medskip

\begin{lemma}
\label{lem.per}
For any  configuration $\sigma_{\Lambda} \in \{-1,+1\}^{\Lambda}$, then
\begin{equation}
\label{eqprop3}
H_N ( \sigma_{{\Lambda},N} ) = \left(\frac{N}{2}\right)^2H_{[2\times 2]}(\sigma_{\Lambda,[2\times 2]}).
\end{equation}

\end{lemma}

\begin{proof}
Note that the number of sites in $\widetilde{\TT}^{[2\times 2]}_N$ is equal to $(N/2)^2$. Therefore, for the external field part in the Hamiltonian we have
\begin{equation}
\begin{array}{lll}
\label{ext.per}
 \displaystyle\sum_{s \in \mathbb{T}_N} h(s) \sigma_{{\Lambda},N}(s) &=&   \displaystyle\sum_{\mathbf{t} \in \wt\TT^{[2\times 2]}_N} \Bigl(  \displaystyle\sum_{s \in \Lambda^{[2\times 2]} + \mathbf{t}} h(s) \sigma_{{\Lambda},N}(s) \Bigr) \\
 &=& \displaystyle\left(\frac{N}{2}\right)^2  \displaystyle\sum_{s \in \mathbb{T}^{[2\times 2]} } h(s) \sigma_{{\Lambda},[2\times 2]}(s).
\end{array}
\end{equation}
For the interaction part   we can write 

{\small \begin{equation}\label{4.66}
\begin{array}{ll}
  \displaystyle\sum_{\langle t, s \rangle \in \mathbb{T}_N} \sigma(t) \sigma(s) = & {} \displaystyle\sum_{\mathbf{t} \in \widetilde\TT_N^{[2\times 2]}}\Bigl( \displaystyle\sum_{\langle t, s \rangle \in \Lambda^{[2\times 2]} + \mathbf{t}}  \sigma(t) \sigma(s)      +\\
 &     \displaystyle\sum_{\substack{\langle t, s \rangle:\:  t \in (\Lambda^{[2\times 2]} + \mathbf{t}), \\ s \in (\Lambda^{[2\times 2]} + \mathbf{t} + (2L_1,0)) }}  \sigma(t) \sigma(s)        +       \displaystyle\sum_{\substack{\langle t, s \rangle:\:  t \in (\Lambda^{[2\times 2]} + \mathbf{t}), \\ s \in (\Lambda^{[2\times 2]}+ \mathbf{t} + (0,2L_2)) }}   \sigma(t) \sigma(s)  \Bigr),
\end{array}
\end{equation}}
for any $\sigma \in \Omega_N$.
We apply this representation for configuration $\sigma_{\Lambda,N}$, and remembering the periodicity condition \eqref{periodicity2} we obtain
\begin{equation}
\label{int.per2}
\begin{array}{lll}
 \displaystyle\sum_{\langle t, s \rangle \in \mathbb{T}_N} \sigma_{{\Lambda},N}(t) \sigma_{{\Lambda},N}(s) & =  &  \displaystyle\sum_{\mathbf{t} \in \widetilde{\mathbb{T}}_N^{[2\times 2]} }  \displaystyle\sum_{\langle t, s \rangle \in \mathbb{T}^{[2\times 2]} }  \sigma_{{\Lambda},[2\times 2]}(t) \sigma_{{\Lambda},[2\times 2]}(s)\\[0.4cm]
 &=&   \displaystyle\left(\frac{N}{2}\right)^2  \displaystyle\sum_{\langle t, s \rangle \in \mathbb{T}^{[2\times 2]}}  \sigma_{{\Lambda},[2\times 2]}(t) \sigma_{{\Lambda},[2\times 2]}(s).
\end{array}
\end{equation}
\medskip
The relation \reff{eqprop3} follows from \eqref{ext.per} and \eqref{int.per2}.  \end{proof}

Since, the number of sites in $\widetilde{\mathbb{T}}_N$ is equal to $N^2$, using Lemma~\ref{lem.per} we obtain from \eqref{eqprop2}
\begin{equation}
\label{eqprop4}
\mathfrak{z}_{\beta,N} ( \mathcal{B} ( \sigma_{\Lambda}) ) \le \exp\Bigl( -\frac{\beta}{4} \left[H_{[2\times 2]}(\sigma_{\Lambda,[2\times2]}) - H_{[2\times 2]}(\sigma^+_{\Lambda,[2\times 2]})\right] \Bigr),
\end{equation}
for any $\sigma_{\Lambda} \in R(\Lambda)$, where $\sigma^+_{\Lambda,[2\times 2]}$ is $+1$ constant configuration on $\mathbb{T}^{[2\times 2]}$.

Finally the Proposition~\ref{prop1} readily follows from the next lemma. The proof of this lemma is essentially repeat the arguments of Lemma 3.3 from \cite{MPS}, but we provide the proof for completeness.

\medskip

\begin{lemma}
\label{lem.hb}
If $h$ satisfies \reff{Peierls2} then for all configuration $\sigma_{\Lambda} \in R(\Lambda)$, its extended configuration 
$\sigma_{\Lambda,[2\times 2]}$, defined by \reff{periodicity2}, on torus $\mathbb{T}^{[2\times 2]}$ satisfies
\begin{equation}
\label{ham.block1}
H_{[2\times 2]}(\sigma_{\Lambda,[2\times 2]}) - H_{[2\times 2]} (\sigma^+_{\Lambda,[2\times 2]}) \ge 4\left(2J - h \frac{L_1L_2}{L_1+L_2}\right).
\end{equation}
\end{lemma}

\begin{proof} 
The proof is taken from \cite{MPS}. For brevity of notations, proving \reff{ham.block1} we omit the indices $\Lambda,[2\times 2]$ everywhere for the configurations from  $\mathbb{T}^{[2\times 2]}$. Let for  $V \subset \mathbb{T}^{[2\times2]}$ the configuration $\sigma^V$ be the perturbation of $\sigma$, that is
\begin{equation}
\label{perturbation2}
\sigma^V(t)=
\left\{
\begin{array}{rl}
-\sigma(t),&\text{ if }t\in V,\\[0.2cm]
\sigma(t),&\text{ if }t\notin V.
\end{array}
\right.
\end{equation}
Let denote $\sigma^{V,+}$ the perturbation of $\sigma^+$. 

We prove the inequality \eqref{ham.block1} for all $V\neq \emptyset$ and $V\neq  \mathbb{T}^{[2\times 2]}$. Since $V\subset\mathbb{T}^{[2\times 2]}$ there are the edges $\langle u,v\rangle\subset\mathbb{T}^{[2\times 2]}$ such that $u\in V,v\in V^c$, thus $\sigma^{V,+}(u)=-1,\:\sigma^{V,+}(v)=+1$. The set of those edges composes {Peierls contour} (or simply, a contour) $\partial V$ of $V$.  Peierls contour is union of the horizontal edges $\partial^h V$ and the vertical edges $\partial^v V$, i.e. $\partial V=\partial^h V\cup \partial^v V$. Note that,
\begin{equation}
\label{ham.vert}
H_{[2\times 2]}(\sigma^{V,+}) - H_{[2\times 2]} (\sigma^+) = 2 J |\partial V | + 2 \displaystyle\sum_{s \in V} h(s).
\end{equation}
We represent $V$ as union of   sets $\mathcal S_V=\{S\}$ of connected horizontal lines of the sites or as union of sets $\mathcal T_V=\{T\}$ of connected vertical lines of the sites. Then
\begin{equation}
\label{ham.vert2}
\displaystyle\sum_{s \in V} h(s) =   \displaystyle\sum_{S \in \mathcal{S}_V}\sum_{s \in S} h(s)=  \displaystyle\sum_{T \in \mathcal{T}_V}\sum_{s \in T} h(s).
\end{equation}

 There are two types of the connected lines of the sites in $V$: \textit{closed} or \textit{open}. In the last case the line has two ends which belong to the boundary edges of $\partial V$. If the open line is horizontal then there exist two edges of $\partial^h V$ having common points with the line. The same property is true for the vertical open lines intersecting $\partial^v V$. The closed lines do not intersect the boundary $\partial V$. Their are closed circles around the torus $\TT^{[2\times 2]}$.

Let $\mathcal{S}_V=\mathcal{S}_V^{cl}\cup\mathcal{S}_V^{op}$, where $\mathcal{S}_V^{cl}$ is the subset of  the closed horizontal lines and $\mathcal{S}_V^{op}$ is the subset of  the open horizontal lines. The similar representation holds for the vertical lines, $\mathcal{T}_V=\mathcal{T}_V^{cl}\cup\mathcal{T}_V^{op}$. For the closed lines
\begin{equation}
\label{hamclosed}
\sum_{s \in S} h(s)   =  \sum_{s \in T} h(s) =  0,
\end{equation}
where $S\in\mathcal{S}_V^{cl}$ and $T\in\mathcal{T}_V^{cl}$. For the open lines the following estimates hold
\begin{equation}
\label{ham.vert3}
\begin{array}{lll}
\Bigl| \displaystyle\sum_{s \in S} h(s) \Bigr| & = &  h \bigl| | S\cap \mathbf{Z}_+| -  | S\cap \mathbf{Z}_-| \bigr| \le hL_1,\\[0.3cm]
\Bigl| \displaystyle\sum_{s \in T} h(s) \Bigr| & = & h \bigl| | T\cap \mathbf{Z}_+| -  | T\cap \mathbf{Z}_-| \bigr| \le hL_2,
\end{array}
\end{equation}
where $S\in\mathcal{S}_V^{op}$ and $T\in\mathcal{T}_V^{op}$. This implies the following inequalities
\begin{equation}
\label{ham.vert4}
\begin{array}{l}
2 \displaystyle\sum_{S \in \mathcal{S}_V}\sum_{s \in S} h(s)  \ge  -2 h L_1 |\mathcal{S}_V| \ge -hL_1 |\partial^h V|,\\[0.3cm]
2 \displaystyle\sum_{T \in \mathcal{T}_V}\sum_{s \in T} h(s)  \ge  -2 h L_2 |\mathcal{T}_V| \ge -hL_2 |\partial^v V|.
\end{array}
\end{equation}
Finally, by \reff{ham.vert}, \reff{ham.vert2} and \reff{ham.vert4},
{\small \begin{equation*}
\label{ham.vert5}
\begin{array}{lll}
H_{[2\times 2]}(\sigma^{V,+}) - H_{[2\times 2]} (\sigma^+) &=& 2 J |\partial V | + \displaystyle\frac{L_1}{L_1+L_2} \Bigl(2 \displaystyle\sum_{s \in V} h(s) \Bigr)+ \frac{L_2}{L_1+L_2} \Bigl(2 \displaystyle\sum_{s \in V} h(s) \Bigr)\\[0.3cm]
& \ge & 2 J |\partial V | -h \displaystyle\frac{L_1L_2}{L_1+L_2} |\partial^h V| -h \displaystyle\frac{L_1L_2}{L_1+L_2} |\partial^v V|\\[0.3cm]
& = & \displaystyle \Bigl(2J - h \frac{L_1L_2}{L_1+L_2}\Bigr) |\partial V|.
\end{array}
\end{equation*}}
Clearly for each $V\subset \mathbb{T}^{[2]}$, $|\partial V| \ge 4$. We have proved \reff{ham.block1}.  \end{proof}

Finally, applying this lemma to \eqref{eqprop4}, and by sub-additivity of $\mathfrak{z}_{\beta, N}$ (see (\ref{subz})), we have
\begin{equation}
\label{z.add.lem3}
\begin{array}{lll}
\mathfrak{z}_{\beta, N}(\mathcal{R}) & \le & \displaystyle\sum_{\sigma_{\Lambda} \in R(\Lambda)} \mathfrak{z}_{\beta, N}( \mathcal{B}(\sigma_{\Lambda})  ) \\
& \le & \displaystyle\sum_{\sigma_{\Lambda} \in R(\Lambda)} \exp \left( -\beta \Bigl(2J - h\frac{L_1L_2}{L_1+L_2} \Bigr) \right)\\%[0.6cm]
&= &|R(\Lambda)| \displaystyle\exp \left( -\beta \Bigl(2J - h\frac{L_1L_2}{L_1+L_2} \Bigr) \right).
\end{array}
\end{equation}
The number $|R(\Lambda)|$ of the bad configurations on $\Lambda$ is estimated as $|R(\Lambda)|\leq 2^{B_1B_2}$. It proves \reff{4.28}.

The inequality \reff{4.30} we obtain by the same way as \reff{4.28} was proved, using the estimate $|R(\Lambda^*)|\leq 4^{B_1B_2}$.          \qed

\subsubsection{Proof of Proposition~\ref{twopoints}.}
Denote $\Omega_N^{s,t}= \{ \sigma \in \Omega_N: \sigma(s)=+1 \text{ and } \sigma(t)=-1\}$. Then for each $\sigma \in \Omega_N^{s,t}$ we define the set
\begin{equation}
\label{pluscluster}
I^+(\sigma) = \left\{ u \in \mathbb{T}_N : \sigma(u) = +1\right\},
\end{equation}
and let $I^+(\sigma,s) \subseteq I^+(\sigma)$ be its maximal connected component containing the site $s$.  The sites $s$ and $t$ are separated by some Peierls contour $\gamma(\sigma)$.  As above  Peierls contour is  the set of edges $\{\langle u,v\rangle\}$ such that $\sigma(u)\ne\sigma(v)$. 

We define a dual to $\gamma(\sigma)$ contour $\gamma^*(\sigma)$. To this end we consider the dual lattice $\Z^{*2}$ and a dual graph $\mathbb G^*=(\Z^{*2},\mathbb E^*)$ . The edges $\langle u^*,v^*\rangle\in\mathbb E^*$ are orthogonal to the edges from $\mathbb E$ (see \reff{graphZ2}). The dual to $\gamma(\sigma)$ contour $\gamma^*(\sigma)=\langle u^*,v^*\rangle$ consists of the all dual edges which are orthogonal to the edges from $\gamma(\sigma)$.

More formal description is as following. An edge $\langle u^*,v^*\rangle$  with $u^*=(u_1^*,u_2^*)$ and $v^*=(v_1^*,v_2^*)$ dual to $\langle u,v\rangle$  with $u=(u_1,u_2)$ and $v=(v_1,v_2)$, is defined by
\begin{equation*}
\begin{array}{ll}
u_1^*=&\frac{u_1+v_1}{2}|u_1-v_1|+(u_1+\frac12)|u_2-v_2|,\\
u_2^*=&\frac{u_2+v_2}{2}|u_2-v_2|+(u_2+\frac12)|u_1-v_1|,\\
v_1^*=&\frac{u_1+v_1}{2}|u_1-v_1|+(v_1-\frac12)|u_2-v_2|,\\
v_2^*=&\frac{u_2+v_2}{2}|u_2-v_2|+(v_2-\frac12)|u_1-v_1|.
\end{array}
\end{equation*}
Let $\gamma^{ext}(\sigma,s)\subseteq\gamma^{*}(\sigma)$ be such that any point $r\in\gamma^{ext}(\sigma,s)$ can be connected with the site $t$ by a line in $\wh\TT_N$ avoiding $\wh I^+(\sigma,s)$ (see \reff{torocont}). Let $\wh J^+(\sigma,s)\subset\wh\TT_N $ contain $s$ and only its boundary is $\gamma^{ext}(\sigma,s)$. It is clear that $\wh J^+(\sigma,s)\supset \wh I^+(\sigma,s)$. The contour $\gamma^{ext}(\sigma,s)$ is called an external contour relatively to the site $s$.

Denote $\Gamma_{s,t}=\{\gamma^{ext}(\sigma,s):\:\sigma\in\Omega_N^{s,t}\}$ the set of all external contours. 

In what follows $\Psi$ means either the block or the double-block. Recall that $\Psi$ is  a subgraph of \reff{graphZ2} embedded in $\R^2$. Further, for a given $L_1$ and $L_2$ we define a set $\mathfrak P=\{\Psi\}$ of permissible blocks which will participate in definition of a thick block-contour. First, all $\Lambda$-blocks belong to the set $\mathfrak P$. Second, as defined in \reff{alllambdas}, if $L_1$ is even, then all the horizontal $\Lambda_{h,1}^*$- and $\Lambda_{h,2}^*$-blocks belong to $\mathfrak P$; and if $L_2$ is even, then all vertical $\Lambda_{v,1}^*$- and $\Lambda_{v,2}^*$-blocks belong to $\mathfrak P$ as well.

Note that if $L_1$ and $L_2$ are odds, then the set  $\mathfrak P$ consists on only $\Lambda$-blocks. Note also that any double block $\Psi$ from $\mathfrak P$ consists of two disjoint $\Lambda$-blocks, say $\Lambda'$ and $\Lambda''$. We associate any such double block with a set $\Phi_\Psi \subset\Psi$ of edges connecting $\Lambda'$ and $\Lambda''$: $\Phi_\Psi = \{ \langle u, v \rangle: \ u\in \Lambda' \mbox{ and } v \in \Lambda'' \}$.

We say that a $\Lambda$-block $\Psi$ from $\mathfrak P$ intersects a contour $\gamma^{ext}$, $\Psi\cap\gamma^{ext}\ne\emptyset$, if there exists an edge $\langle u, v \rangle, u,v \in \Psi$ such that $\langle u^*, v^* \rangle \subset \gamma^{ext}$.
We say that a $\Lambda^*$-block $\Psi$ from $\mathfrak P$ (if it exists) intersects a contour $\gamma^{ext}$, $\Psi\cap\gamma^{ext}\ne\emptyset$, if 
$\gamma^{ext}$ intersects all edges from $\Phi_\Psi$.

For any $\gamma^{ext} \in \Gamma_{s,t}$ we define $\mathcal E=\mathcal E(\gamma^{ext})=\left\{\Psi:\:\Psi\cap\gamma^{ext}\ne\emptyset\right\}$. The set $\mathcal E=\mathcal E(\gamma^{ext})$ is called a \textit{thick external contour} corresponding Peierls€™ external contour  $\gamma^{ext}$. Let $\mathcal E_b\subseteq\mathcal E$ be the subset of the $\Lambda^*$-blocks, and let $\mathcal E_0=\mathcal E\setminus\mathcal E_b$ the set of $\Lambda$-blocks.

Denote $\Omega^*(\mathcal E)$, the set of configurations $\sigma$, that generate the thick external contour $\mathcal E$. If $\sigma\in \Omega^*(\mathcal E)$ then $\sigma\in \bigcap_{\Psi\in\mathcal E_0}\mathcal B(\sigma(\Psi) )\cap \bigcap_{\Psi^*\in\mathcal E_b}\mathcal{B}(\sigma(\Psi^*))$. Moreover $\sigma\in \bigcap_{\Psi\in\mathcal E_0}\mathcal R_\Psi\cap\bigcap_{\Psi^*\in\mathcal E_b}\mathcal R_{\Psi^*}$, that is 
\begin{equation}
\label{inclusao}
\Omega^*(\mathcal E)\subset \bigcap_{\Psi\in\mathcal E}\mathcal R_\Psi.
\end{equation}
Next we estimate $\mu_{\beta,N}(\Omega^*(\mathcal E))$. First, we separate the set of $\Lambda^*$-blocks in all the possible $\Lambda_k^*$-blocks, $k\in\mathcal{D}$. Let $\mathcal E_b = \mathcal E_{h,1} \cup  \mathcal E_{h,2}  \cup  \mathcal E_{v,1} \cup  \mathcal E_{v,2}$, the subsets of $\Lambda^*$-blocks defined on the subgroups \reff{alltori}. Using \reff{inclusao}, the Cauchy-Schwarz inequality, and the chessboard estimates \reff{chessboard1} and \reff{chessboard1_1}, respectively, we obtain

{\small \begin{equation}\label{4.84}
\begin{array}{ll}
\mu_{\beta,N}(\Omega^*(\mathcal E))\leq\mu_{\beta,N}(\bigcap_{\Psi\in\mathcal E}\mathcal R_\Psi)\leq\sqrt{\mu_{\beta,N}\left(\bigcap_{\Psi\in\mathcal E_0}\mathcal R_\Psi\right)}{\displaystyle\prod_{k \in\mathcal D}} \sqrt[8]{\mu_{\beta,N}\left(\bigcap_{\Psi^*\in\mathcal E_{k}}\mathcal R_{\Psi^*}\right)}\\
\leq\prod_{\Psi\in\mathcal E_0} \Bigl[\mu_{\beta,N} \Bigl( \bigcap_{\tb\in\wt\TT_N}\pi_\tb(\mathcal R_\Psi)\Bigr) \Bigr]^{\frac1{2N^2}}
{\displaystyle\prod_{k \in\mathcal D}}\prod_{\Psi^*\in\mathcal E_{k}}\Bigl[\mu_{\beta,N}\Bigl(\bigcap_{\tb\in\wt\TT_N^{(k)}}\pi_\tb(\mathcal R_{\Psi^*})\Bigr)\Bigr]^{\frac1{4N^2}}.
 \end{array}
\end{equation}}

Recall that $N^2$ is the number of the $\Lambda$-blocks and $N^2/2$ is the number of the $\Lambda^*$-blocks. The probabilities $\mu_{\beta,N}\bigl(\bigcap_{\tb\in\wt\TT_N}\pi_\tb(\mathcal R_\Psi)\bigr)$ when $\Psi\in\mathcal E_0$, and $\mu_{\beta,N}\bigl(\bigcap_{\tb\in\wt\TT_N^{(k)}}\pi_\tb(\mathcal R_{\Psi^*})\bigr)$ when $\Psi^*\in\mathcal E_b$ do not depend on the position of $\Psi$ ($\Psi^*$) thus we introduce a magnitude $\mu_{\beta,N}\bigl(\bigcap_{\tb\in\wt\TT_N}\pi_\tb(\mathcal R)\bigr)$ meaning the probability of any propagated block. Moreover, we define the magnitude $\mu_{\beta,N}\bigl(\bigcap_{\tb\in\wt\TT_N^{(k)}}\pi_\tb(\mathcal R^*)\bigr)$ meaning the probability of any propagated double-block. Using \reff{4.84} we obtain

\begin{equation*}
\mu_{\beta,N}(\Omega^*(\mathcal E))\leq\Bigl[\mu_{\beta,N}\Bigl(\bigcap_{\tb\in\wt\TT_N}\pi_\tb(\mathcal R)\Bigr)\Bigr]^{\frac{|\mathcal E_0|}{2N^2}}{\displaystyle\prod_{k \in\mathcal D}} \Bigl[\mu_{\beta,N}\Bigl(\bigcap_{\tb\in\wt\TT_N^{(k)}}\pi_\tb(\mathcal R^*)\Bigr)\Bigr]^{\frac{|\mathcal E_k |}{8(N^2/2)}}.
\end{equation*}
Using notation \reff{zeta} and \reff{zeta1_1}, and applying \reff{4.28} and \reff{4.30}, we obtain
\begin{equation}\label{4.85}
\begin{array}{lll}
\mu_{\beta,N}(\Omega^*(\mathcal E))&\leq&\mathfrak z_{\beta,N}(\mathcal R)^{\frac{|\mathcal E_0|}{2}}\prod_{k \in\mathcal D}\mathfrak z_{\beta,N}(\mathcal R^*)^{\frac{|\mathcal E_k |}{8}}\\[0.3cm]
&\leq& 2^{B_1B_2\frac{|\mathcal E_0|}{2}} 4^{B_1B_2\frac{|\mathcal E_b|}{8}}\exp\left\{-\frac{\beta}{8}|\mathcal E|\left(2J-\frac{hL_1L_2}{L_1+L_2}\right)\right\}.
\end{array}
\end{equation}

The next step is defining $\mathfrak D_{s,t}=\{\mathcal E=\mathcal E(\gamma^{ext}), \gamma^{s,t}\in\Gamma_{s,t}\}$.
Next, we estimate $\mu_{\beta,N}(\Omega_N^{s,t})$ using the inclusion $\Omega_N^{s,t}\subset\bigcup_{\mathcal E\in\mathfrak D_{s,t}}\Omega^*(\mathcal E)$. Considering the inequality \reff{4.85}, and a combinatorial argument, we obtain

{\small \begin{equation}\label{4.86}
\begin{array}{ll}
\displaystyle\mu_{\beta,N}(\Omega_N^{s,t})\leq\sum_{\mathcal E\in\mathfrak D_{s,t}}\mu_{\beta,N}(\Omega^*(\mathcal E))\leq \sum_{n\geq1}\sum_{\stackrel{n_0,n_b:}{n_0+n_b=n}}c^{n_0}c^{2n_b}(2^{n_0}\sqrt{2}^{n_b})^{B_1B_2/2}e^{-\beta n\alpha/8}\\
= \displaystyle\sum_{n\geq1}\left(c2^{B_1B_2/2}+c^24^{B_1B_2/8}\right)^ne^{-\beta n\alpha/8} \leq \displaystyle\sum_{n\geq1}\bigl(c(c+1)2^{B_1B_2/2}e^{-\beta\alpha/8}\bigr)^n,
\end{array}
\end{equation}}
where $\alpha=2J-\frac{hL_1L_2}{L_1+L_2}$, $n_0=|\mathcal E_0|$, $n_b=|\mathcal E_b|$, and $c$ is a combinatorial constant related to the number of the thick contours. The way the constant $c$ appears in \reff{4.86} defined by our calculations the number of the thick contour having its length equal to $n$ (see a justification  below). If \reff{betaP3} holds true, then,
\begin{equation}
\mu_{\beta,N}(\Omega_N^{s,t})\leq 2c(c+1)2^{B_1B_2 /2}\exp\left\{-\frac{\beta}{8} \left(2J - \frac{hL_1L_2}{L_1+L_2}\right)\right\}.
\end{equation}
\qed

\medskip

\textbf{Justification of \reff{4.86}}. The double-blocks create two neighboring vertices in the graph. Therefore the double-block is taking in account twice in \reff{4.86}. Thus the contribution of the double-block energy is estimated as $c^{2n_b}4^{n_bB_1B_2}\exp\{-\beta\alpha n_b\}$. %It is a justification of \reff{4.86} that concerns the energy of the double-blocks. 
\qed

%%%%%%%%%%%%%%%%%%%%%%%%%%%%
%%%%%%%%%%%%%%%%%%%%%%%%%%%%
%%%%%%%%%%%%%%%%%%%%%%%%%%%%

\section{2D ferromagnetic Ising model with alternating strips external field}
\label{sec:nardi}

In \cite{NOZ} it was studied a phase diagram of the 2D Ising model with alternating external field on 1D sublattices. In this section we prove the result of \cite{NOZ} about coexistence of two phases by using RP in a way similar to the considerations in previous sections. In fact we prove a more general result of the coexistence, including the coexistence result of \cite{NOZ}.
The model in \cite{NOZ} is as follows. The external field is

\begin{equation}
\label{ext.nar}
h(s_1, s_2)=
\left\{
\begin{array}{rl}
h,&\text{ if }s_2 \text{ is even},\\
-h,&\text{ if }s_2 \text{ is odd},
\end{array}
\right.
\end{equation}
where $h >0$, and the Hamiltonian is defined by \eqref{ham.chess}.

Thus this model is an ``extreme" case of the cell-board model. Indeed, the external field in \reff{ext.nar} can be obtained letting $L_1=\infty$ and $L_2=1$.
In \cite{NOZ}, it is proved that a phase transition in this model holds true for $\beta$ sufficiently large and $h< 2J - k e^{-\beta J}$, where $k$ being a suitable positive constant.
We propose a more general model, with $L_2 \ge 1$, for which we use the reflection positivity techniques to prove phase transition. We called this model, {\it the Ising model with alternating strips external field}.

Formally, consider $L \in \mathbb{N}$, for each integer $n$, we define a {\it strip} of size $L$ by

\begin{equation}
\label{faixaL}
F(n) =  \{ (t_1,t_2) \in \Z^2 :  nL \le t_2 < (n+1)L \}.
\end{equation}

Note that $\mathbb Z^ 2 = \cup_{n\in\mathbb Z} F(n)$. In a similar way to \reff{Zplus}, we think the strips being colored black or white. See Figure~\ref{blocknardi}, where $L=2$. Then, we define the subset $\mathbf{Z}^*_+$ and $\mathbf{Z}^*_-$ of $\mathbb{Z}^2$.

\begin{equation}
\begin{array}{llr}
\label{ZplusL}
\mathbf{Z}^*_+ = \displaystyle\bigcup_{\substack{n:\\ n \text{ is even}}} F(n), & & \hfill \mathbf{Z}^*_- = \Z^2 \setminus \mathbf{Z}^*_+.
\end{array}
\end{equation}

Let $\Omega = \{-1,+1\}^{\Z^2}$ the set of all configurations on $\Z^2$. The formal Hamiltonian, for the Ising model with alternating strips, is defined by \reff{ham.chess}, and the external field is given by
\begin{equation}
\label{extfaixas}
h(s)=
\left\{
\begin{array}{rl}
h,&\text{ if }s\in \mathbf{Z}^*_+,\\
-h,&\text{ if }s\in \mathbf{Z}^*_-,
\end{array}
\right.
\end{equation}
where $h>0$. For this model, we obtain the following result.

\begin{theorem}
\label{ournardi}

Consider the Ising model defined by the Hamiltonian \reff{ham.chess}, with external field given by \reff{extfaixas} and \reff{ZplusL}. If $h < 2J/L$, then there exists a suitable positive constant $k = k(L)$, such that for any $\beta > k / ( 2J-hL)$, there exist two distinct measures $\mu^+_{\beta}$ and $\mu^-_{\beta} \in \mathcal{G}_{\beta}$, which satisfy
\begin{equation}
\label{Ntheorem1}
\begin{array}{l}
\mu^{\pm}_{\beta} (\sigma(t) = \pm 1 ) > \frac{1}{2}.
\end{array}
\end{equation}\end{theorem}

\begin{proof}
We follow ideas of section \ref{sec:RP}. We remark that in this section, we shall use the same notation, without distinction with previous model.

First, we construct a torus $\mathbb{T}_N$ by taking a subset $T_N$ of $\mathbb Z^2$ of size $N \times N L$:  $$T_N=\{ t=(t_1,t_2)\in \mathbb Z^2: \ 0\le t_1 < N, 0\le t_2 < NL\},$$

\noindent where $N$ is multiple of 4. Thus, $\mathbb T_N$ is the factor-group $\mathbb Z / (N \mathbb Z) \times \mathbb Z / (NL \mathbb Z)$. We consider the corresponding Hamiltonian with periodical boundary condition, as defined in \eqref{torusham}. 

Similar to \reff{planes}, we define the sets of planes $\mathcal P= \mathcal P_1 \cup \mathcal P_2$, given by

\begin{equation}
\label{planes2}
\begin{array}{c}
P_1^{(n)} = \{ t=(t_1,t_2)\in \mathbb R^2: \ t_1 = n\},\\[0.3cm]
P_2^{(n)} = \{ t=(t_1,t_2)\in \mathbb R^2: \ t_2 = nL + (L-1)/2\},
\end{array}
\end{equation}
where $n<N$ a positive integer, and $\mathcal P_1=\{P_1^{(n)}\} $, and $\mathcal P_2=\{P_2^{(n)}\} $.

The set of planes $ \mathcal{P}$, decompose the torus $ \mathbb{T}_N$ in rectangular blocks (see Figure~\ref{blocknardi}). It is easy to check \reff{criterion} for all planes $P \in \mathcal{P}$, using same functions as in the proof of Proposition \ref{ourisRP}. Once the PR property is guaranteed, we apply the chessboard estimates (see \reff{chessboard1}). For this reason, we define the $ \Lambda$-blocks by translations of the block $\Lambda$. We construct $ \Lambda $, considering on $\mathbb R^2$ the rectangle $$\tilde\Lambda = 
\Bigl\{ (t_1,t_2): \Bigl| t_1 + \frac{1}{2} \Bigr| \le \frac{1}{2},\ \Bigl| t_2 + \frac{1}{2} \Bigr| \le \frac{L}{2} \Bigr\}.$$

Therefore, $\Lambda = \tilde\Lambda \cap \mathbb Z^2$. In other words, $\widetilde{\mathbb{T}}_N = \{ \mathbf{t}=(t_1,t_2) \in \mathbb{T}_N : t_1 = n, t_2 = m L, n, m \in \Z \}$.

Now, we use the definition of bad-block event (see \eqref{badconf} for the definition of $R(\Lambda)$, and \eqref{setbadblock} for the event $\mathcal{R}$). Estimating $\mathfrak{z}_{\beta, N} (  \mathcal{R})$, note that Lemma~\ref{lem.per} still holds for this model. However, Lemma~\ref{lem.hb} should be re-written. Therefore, we will explain the steps to follow to obtain the desired estimate. From inequality \reff{ham.vert3}, we consider $| \sum_{s \in T} h(s) |  \le hL$, for each $T\in\mathcal{T}_V$. That implies, an equivalent to \reff{ham.vert4}. Remember that, $ |\partial V| =  |\partial^h V| +  |\partial^v V| \ge 4$. Finally, we have
\begin{equation}
H_{[2\times 2]}(\sigma^{V,+}) - H_{[2\times 2]} (\sigma^+)  \ge  2 J |\partial V | -h L |\partial^v V| \ge  4 (2 J -h L) .
\end{equation}

Similar to \reff{z.add.lem3}, we say that $\mathfrak{z}_{\beta, N}(\mathcal{R})  \le |R(\Lambda)| \exp ( -\beta [2J - hL ] )$. We estimate $|R(\Lambda)|\leq 2^{2L^{\prime}}$. In this case, $L^{\prime} = L$, when $L$ is even, and $L^{\prime} = L+1$, if $L$ is odd.

Therefore, as in Proposition~\ref{twopoints}, for every $ s, t \in \mathbb{T}_N$, the following inequality
\begin{equation}
\label{final2}
 \mu_{\beta,N} \left( \sigma(s) = +1, \sigma(t) = -1 \right) \le 2 c(c+1) 2^{2L^{\prime}}\exp \left\{-\frac{\beta}{2} \Bigl(2J-hL \Bigr) \right\},
\end{equation}
holds for all $\beta>k/(2J-hL)$. Here, the constant $ c $ is the same.

As a final step, we repeat considerations with conditional measures \reff{measures}, where $t \in \mathbb{T}_N$, such that $t = s + (N/2,0)$. Note that \reff{prova1a} still holds. Checking that, consider some configuration $\sigma \in \Omega_N$, such that $\sigma(\Lambda) \equiv +1$. So we can construct $\sigma^{\prime} = - \theta_{Q_2}(\sigma) \in \Omega_N$, see equation \reff{4.12}. Then, $H_N(\sigma) = H_N (\sigma^{\prime})$.

Finally, when $N \to \infty$, we obtain the infinite-volume Gibbs measures, $\mu^{+}_{\beta}$ and $\mu^{-}_{\beta}$. Therefore, by estimative to $\mathfrak{z}_{\beta, N}$ above and \reff{final2}, there exists a constant $ k> 0 $, such that for all $\beta>k/(2J-hL)$, \reff{Ntheorem1} holds true.
\end{proof}

%%%%%%%%%%%%%%%%%%%%%%%%%%%%%%%%%%%%%%%%%%%%%%%%%%%%%%%%%%%%%%%%%%%%%%%%

\bigskip

\textbf{Acknowledgments}: Manuel Gonz\'alez Navarrete was supported by BecasChile, Comisi\'on Nacional de Investigaci\'on Cient\'i{}fica y Tecnol\'ogica. The research of Eugene Pechersky was carried out at the IITP-RAS at the expense of the Russian Foundation for Sciences (project No 14-50-00150). Anatoly Yambartsev thanks Conselho Nacional de Desenvolvimento Cien\-t\'i\-fi\-co e Tecnol\'ogico, CNPq (grant 307110/2013-3)  and Funda\c{c}\~ao de Amparo \`a Pesquisa do Estado de S\~ao Paulo, FAPESP (grant 2009/52379-8). The final publication is available at Springer via http://dx.doi.org/10.1007/s10955-015-1392-9

\end{document}